\documentclass[noinfoline]{imsart}
 
 \usepackage{color,hyperref}
    \catcode`\_=11\relax
    \newcommand\email[1]{\_email #1\q_nil}
    \def\_email#1@#2\q_nil{%
      \href{mailto:#1@#2}{{\emailfont #1\emailampersat #2}}
    }
    \newcommand\emailfont{\sffamily}
    \newcommand\emailampersat{{\color{blue}\small@}}
    \catcode`\_=8\relax  

\RequirePackage[OT1]{fontenc}
\RequirePackage{amsthm,amsmath}
\usepackage[square,numbers]{natbib}

\RequirePackage[colorlinks,citecolor=blue,urlcolor=blue]{hyperref}

\usepackage{verbatim}

\allowdisplaybreaks

\usepackage{fullpage}
\usepackage{enumitem, hyperref}
\makeatletter
\def\namedlabel#1#2{\begingroup
    #2%
    \def\@currentlabel{#2}%
    \phantomsection\label{#1}\endgroup
}
\makeatother
\DeclareMathAlphabet{\mymathbb}{U}{BOONDOX-ds}{m}{n}

\usepackage{subfig}
\usepackage{layout}

\usepackage{dsfont}
\usepackage[mathscr]{eucal}
\usepackage[toc,page]{appendix}
\usepackage{mathrsfs}
\usepackage{color}
\usepackage{pifont}
\usepackage{bm}
\usepackage{latexsym}
\usepackage{amsfonts}
\usepackage{amssymb}
\usepackage{epsfig}
\usepackage{graphicx}
\usepackage{multirow}

\newtheorem{theorem}{Theorem}

\newtheorem{corollary}{Corollary}

\newtheorem{example}{Example}

\startlocaldefs
\numberwithin{equation}{section}
\theoremstyle{plain}
\endlocaldefs

\begin{document}

\begin{frontmatter}

\title{
Functional Data Analysis with Rough  {Sample} Paths?}

\runtitle{FDA with Rough Sample Paths?}

\begin{aug}
\author{\fnms{Neda} \snm{Mohammadi }\ead[label=e1]{neda.mohammadijouzdani@epfl.ch}} \and
\author{\fnms{Victor M.} \snm{Panaretos}\ead[label=e2]{victor.panaretos@epfl.ch}}


\runauthor{N. Mohammadi  \& V.M. Panaretos}

\affiliation{\'Ecole Polytechnique F\'ed\'erale de Lausanne}

\address{Institut de Math\'ematiques\\ \'Ecole Polytechnique F\'ed\'erale de Lausanne }
 \email{neda.mohammadijouzdani@epfl.ch}, \email{victor.panaretos@epfl.ch}

\end{aug}

\begin{abstract} 
Functional data are typically modeled as sample paths of smooth stochastic processes in order to mitigate the fact that they are often observed discretely and noisily, occasionally irregularly and sparsely. The smoothness assumption is imposed to allow for the use of smoothing techniques that annihilate the noise. At the same time, imposing the smoothness assumption excludes a considerable range of stochastic processes, most notably diffusion processes.  Under perfect observation of the sample paths, such processes would not need to be excluded from the realm of functional data analysis. In this paper, we introduce a careful modification of existing methods, dubbed the ``reflected triangle estimator", and show that this allows for the functional data analysis of processes with nowhere differentiable sample paths, even when these are discretely and noisily observed, including under irregular and sparse designs. Our estimator matches the established rates of convergence for processes with smooth paths, and furthermore attains the same optimal rates as one would get under perfect observation.  Thus, with reflected triangle estimation, the scope of applicability of much of the methodology developed for discretely/irregularly/noisily/sparsely sampled functional data is considerably extended. By way of simulation it is shown that the advantages furnished are reflected in practice, hinting at potential closer links with the field of diffusion inference.

\end{abstract}

\bigskip


\end{frontmatter}

\tableofcontents

\section{Introduction}

Functional Data Analysis (FDA) treats the problem of drawing statistical inferences pertaining to the law of a random element $X$ valued in an appropriate function space $\mathcal{H}$, based on $n$ independent (or weakly dependent) realisations $X_1,...,X_n$ distributed as $X$. The ambient function space is usually taken to be $\mathcal{H}=L^2[0,1]$. This allows for a very wide range of random functions $X$. Technical considerations require the minimal assumption that $X$ be almost surely continuous, namely in order to be able to make point-wise sense of $X$ as a stochastic process $\{X(t):[0,1]\rightarrow \mathbb{R}\}$ (see \citet{hsing_theoretical_2015}, p.175). But once this is done, the covariance operator of $X$ 
is assured to be trace-class and one can carry out the programme set out in classical monographs such as \citet{bosq_linear_2000}or \citet{ferraty_nonparametric_2006}. 

These monographs, as well a great number of theory/methods papers assume that one can perfectly observe the complete sample path ( {see, e.g., \citet{dauxois_asymptotic_1982},  \citet{hall_properties_2006}, \citet{hormann_consistency_2013} \cite{hormann_weakly_2010},  \citet{hall_methodology_2007},  \citet{panaretos_cramerkarhunenloeve_2013,panaretos_fourier_2013},  \citet{delaigle_methodology_2012}}). This ``complete observation" framework only requires continuity and encompasses processes that can be non-smooth, even nowhere differentiable, such as diffusion processes. When paths are only observed at finitely many locations, though, the path continuity assumption alone is not deemed enough, particularly when the sampled values are additionally corrupted by noise. It is supplemented by additional smoothness assumptions, which are then leveraged in order to non-parametrically estimate the mean and covariance function, as a first step towards further analysis. For instance, the approach popularised by  \citet{ramsay_functional_2005} assumes that $X$ possesses $\mathcal{C}^2$-paths, and advocates pre-smoothing each sampled path to construct fully observed proxies. \citet{yao_functional_2005}, on the other hand assume a $\mathcal{C}^2$ covariance which they smooth directly, an approach which adapts well to irregular and sparse sampling designs. Such assumptions enable the use of FDA techniques in the more realistic scenario of discrete/noisy observation, but limit its scope to processes with smooth paths.

Our purpose is to demonstrate that it is actually \emph{not necessary} to assume twice (or even once) differentiability of the sample paths in order to carry out a functional data analysis with noisily sampled paths. In particular, we show that one can obtain the same uniform and almost sure rates of convergence for mean and covariance estimation even with nowhere differentiable paths, potentially irregularly/sparsely sampled and with measurement error contamination.  {Our approach is based on a conceptually simple modification of procedures based on smoothing the covariance, e.g. \citet{yao_functional_2005}, \citet{hall_properties_2006-1} and \citet{li_uniform_2010}}. The key observation is that rough but continuous sample paths may well possess a covariance function that is smooth except on the diagonal. Prominent such examples include diffusion processes like Brownian motion, the Brownian bridge, or the Ornstein-Uhlenbeck process, to name but a few. So while the path-smoothing approach will fail, the covariance smoothing approach can still succeed if suitably modified to avoid smoothing over (and thus smearing) the singularity along diagonal. 

While the modification is conceptually simple, its ramifications go beyond a weakening of assumptions. From the point of view of theory, they challenge the tenet that Functional Data Analysis relies on \emph{smoothness} of the functional data. And, in doing so they hint at potential closer links with the area of inference for diffusions and related processes, which does not so far share considerable overlap with functional data analysis.  {From the practical point of view, the modification is easy to implement and enjoys the same theoretical guarantees as the well-established covariance smoothing approach. Thus, it can essentially ``automatically" extend much of the FDA toolbox constructed with covariance smoothing at its foundation, to cover a considerably broader range of processes}.

\indent  {To our best knowledge this is the first study to estimate the covariance structure of processes with rough sample paths in the FDA setting of irregular (perhaps sparse) and noise corrupted measurements of multiple trajectories. In a distinct context, \citet{pos_superconsistent_2020} aimed to detect covariance singularities, in the context of functional regression with ``points of impact". However, as they pointed out, their methodology does not bear on the problem of actually \emph{estimating} the covariance.   
}

\section{Problem and Background}\label{problem:notation}
Let $\mathcal{H} = L^2[0,1]$  be the Hilbert space of squared integrable real-valued functions on the unit interval $[0,1]$, equiped with the following inner product
\begin{eqnarray*}
 \langle f,g \rangle_{\mathcal{H}} = \int_{[0,1]}f(t)g(t)dt;\;\;\; f,g \in  \mathcal{H},
\end{eqnarray*}
and the induced norm $\Vert \cdot \Vert_{\mathcal{H}}$ . Let $X$ be a random element of $\mathcal{H}$ (formally a measurable map between some probability space $\left(\Omega , \mathcal{A}, \mathbb{P} \right)$ and $\mathcal{H}$) with almost surely continuous sample paths. Let $\mu(t)=\mathbb{E}\{X(t)\}$ be its mean and $C (\cdot, \cdot)= \mathrm{Cov}\left\{X(\cdot),X(\cdot)\right\}$ be its covariance function. The estimation of $\mu$ and $C$ is a basic first step involved in most functional data analyses, including regression, prediction, testing, and classification. This is to be done based on $n$  i.i.d curves which are observed only at some (random) nodes, subject to measurement error contamination. Concretely, we observe
\begin{eqnarray}\label{initial model}
Y_{ij}= X_i\left( T_{ij}\right) + U_{ij}; \; i=1,2,\ldots n, \; j=1,2,\ldots r_n,
\end{eqnarray}
where  $\left\{U_{ij}\right\}$ forms an i.i.d. array of measurement errors with finite variance $\sigma^2$, and $\left\{T_{ij}\right\}$ is an array of random design points, assumed without loss of generality to satisfy $T_{ik} < T_{ij}$, for $k <j$. Finally,  $\left\{r_n\right\}$ is  the number of measurements per sample path, allowed to take different values as $n$ grows, accordingly yielding dense ($r_n\rightarrow\infty$ as $n\rightarrow\infty$) or sparse ($r_n<R<\infty$) settings. The elements of the arrays $\left\{X_{i}\right\}$ ,  $\left\{T_{ij}\right\}$ and $\left\{U_{ij}\right\}$ are totally independent across all indices $i$ and $j$.

\smallskip
 {Perhaps the first general approach focusses on the dense regime $(r_n\rightarrow\infty)$ and was popularised by \citet{ramsay_functional_2005}.  It consists in smoothing each sample path individually, i.e. applying a smoother to each of $n$ scatterplots $\{(T_{ij}),Y_{ij}\}_{j=1}^{r_n}$.  Typically spline or local polynomial smoothing is applied, yielding $n$ smoothed proxy curves $\widetilde{X}_i$ in lieu of the true latent ${X}_i(t)$. The covariance of the proxy curves $\{\widetilde{X}_i\}$ is then used as a smooth estimator of the true underlying covariance operator. \citet{hall_properties_2006-1} assert that this proxy covariance enjoys the same convergence rate as of the empirical estimator obtained by fully functional curves $\left\{ X_i \right\}$  with no noise, under sufficiently dense observation.}

\citet{yao_functional_2005} examine a similar problem, but focus on the case where $r_n$ is finite (the sparse case). In such a setting, one cannot smooth each individual sample path. Therefore, differently from the approach of \citet{ramsay_functional_2005}, they opt to estimate the covariance directly by smoothing the ``raw covariance" 2D scatterplot $\left\{((T_{ij}, T_{ik}), Y(T_{ij}) Y(T_{ik})) \right\}_{1 \leq k,j\leq r_n,i\leq n}$. To circumvent identifiability issues near the diagonal (due to the noise contamination), they only retain scatterplot points with non-coincident time stamps,  $\left\{\left(T_{ij},T_{ik} \right) \;|\;  i=1,2,\ldots,n,\;k \neq j\right\}$,  and apply a local polynomial smoother on $[0,1]^2$ to obtain an estimator. Though the procedure was motivated by the sparse case $(r_n<R<\infty)$ it can just as well be applied to the dense case ($r_n\rightarrow\infty$), and \citet{li_uniform_2010} later obtained a complete description of the convergence rates of such a procedure regardless of the sampling regime, making the dependence on $r_n$ explicit.

 {Neither the path-smoothing approach (\citet{ramsay_functional_2005}), nor the covariance-smoothing approach (\citet{yao_functional_2005}) will apply in cases where the sample paths of $X$ are non-smooth, however:}
\begin{itemize}
\item  {An essential requirement for the procedure of \citet{ramsay_functional_2005} is that the sample paths be of class $\mathcal{C}^2$ almost surely. Therefore, there is no hope to use this procedure with paths of low regularity, for example paths that are continuous but may be nowhere differentiable.}

\item  {The procedure of \citet{yao_functional_2005}, requires the covariance to be of class $\mathcal{C}^2$ over $[0,1]^2$, and thus the procedure does not apply to continuous processes that do not have twice (mean-square) differentiable sample paths. This excludes diffusion processes, among others, whose covariance fails to be differentiable along the diagonal.}
\end{itemize}

\medskip
The damage incurred by ``roughness" of the sample paths to the path-smoothing approach seems difficult to repair (and in any case this approach only covers the dense sampling regime). Nevertheless, we show in the next section that a careful modification of the covariance smoothing approach, yields methodology that: on one hand is much more broadly applicable, including to nowhere differentiable processes like diffusion processes; and, on the other hand, attains the same rates of convergence as established by \citet{li_uniform_2010}, with the same explicit dependence on $r_n$ (hence covering both sparse and dense regimes).

 \section{Methodology and Theory}\label{main:result}

 \subsection{Motivation}
 
 
 The main observation behind our approach is that, while non-smooth processes may fail to have a covariance that is $\mathcal{C}^2$ on $[0,1]^2$, they may very well still possess a covariance that is smooth when restricted on the triangle 
 $$\bigtriangleup =  \left\{(s,t)\in [0,1]^2\;|\; 0 \leq t \leq s \leq 1  \right\},$$
and, of course, this restriction determines the overall covariance completely. To see this, consider the case of standard Brownian motion on $[0,1]$, whose sample paths are nowhere differentiable almost surely (and not even differentiable in mean-square). The corresponding covariance is $C(s,t)=\min\{s,t\}$, and is non-differentiable on the diagonal $s=t$. When restricted on $\bigtriangleup$ the covariance becomes $C(s,t)=t$ and is infinitely differentiable (even real analytic). Similar observations can be made for the standard Brownian bridge ($C(s,t)=t-st$ on $\bigtriangleup$), Ornstein-Uhlenbeck process ($C(s,t)=e^{-\beta(t+s)}(e^{2\beta t}-1)/(2\beta)$ on $\bigtriangleup$, where $\beta>0$ is the drift), geometric Brownian motion ($C(s,t)=\exp\{(t+s)/2\})\exp\{t\}-1)$ on $\bigtriangleup$),  { and many more examples.}  { This class of covariance functions also includes isotropic covariance functions, where $C(s,t)$ is a function of $\vert s-t \vert$, that in general admit lower degree regularity on the diagonal segment,  see panel (a) of Figure \ref{surface:matern:BM}.}\\
\noindent These are all instances of a very general setting with a covariance function that is of the form
 $$C (s,t) = g\{\mathrm{min}(s,t) , \mathrm{max}(s,t)\},$$
for some smooth function $g$. When this $g$ is at least of class $\mathcal{C}^2$, or equivalently when $C|_{\triangle}\in \mathcal{C}^2(\bigtriangleup)$ (see condition \ref{2-diff:C}), we can view $C$ as a function with symmetric components on $\bigtriangleup$ and $[0,1]^2\setminus\triangle$, and:
\begin{enumerate}
\item first apply a smoother only to indices from the subset $\left\{\left(T_{ij},T_{ik} \right) \;|\;  i=1,\ldots,n,\;k < j\right\}$ to estimate $C$ restricted to $\triangle$. In other words, the restriction $k<j$ smooth only the scatterplot defined over the ``triangle". 

\item then reflect the estimate on $\triangle$ with respect to the diagonal in order to estimate $C$ over all of $[0,1]^2$. 
\end{enumerate}

 {This approach entirely avoids the irregularity along the diagonal segment. Steps 1 and 2 suggest the term ``reflected triangle estimator" for our approach. By contrast, the traditional ``smoothed square estimator" of \citet{yao_functional_2005} smooths the set $\left\{\left(T_{ij},T_{ik} \right) \;|\;  i=1,\ldots,n,\;k \neq j\right\}$. In doing so it ignores the diagonal itself (which helps with denoising) but still ``reaches over" it, thus smearing the diagonal itself and producing an estimate that is $\mathcal{C}^1$ on the diagonal (and everywhere),  even asymptotically. This bias will especially impact the induced eigenvalue estimates,  as these are intrinsically linked to the area surrounding the diagonal. This effect becomes more transparent if one also concerns the fact  that in practice we only have access to some discretized matrix from of the underlying covariance functions}. 

The examples in Figure \ref{surface:matern:BM} illustrate the situation in two cases.  { In particular case (b) shows that even if we have a flat (nowhere strictly convex) surface on $\bigtriangleup$ this phenomenon is unavoidable.} 
 {In particular Example 1 below  shows that even for case (b) where  we have a flat (nowhere strictly convex) surface on $\bigtriangleup$ this phenomenon is unavoidable. }

The precise definition of our estimator is provided in the next subsection. Subsection \ref{Asymptotic Theory} shows that it inherits the same (optimal) rates as one would have in the smooth case.
   \begin{figure}[t!]%
    \centering
    \subfloat[\centering ]{{\includegraphics[height=.19\textheight]{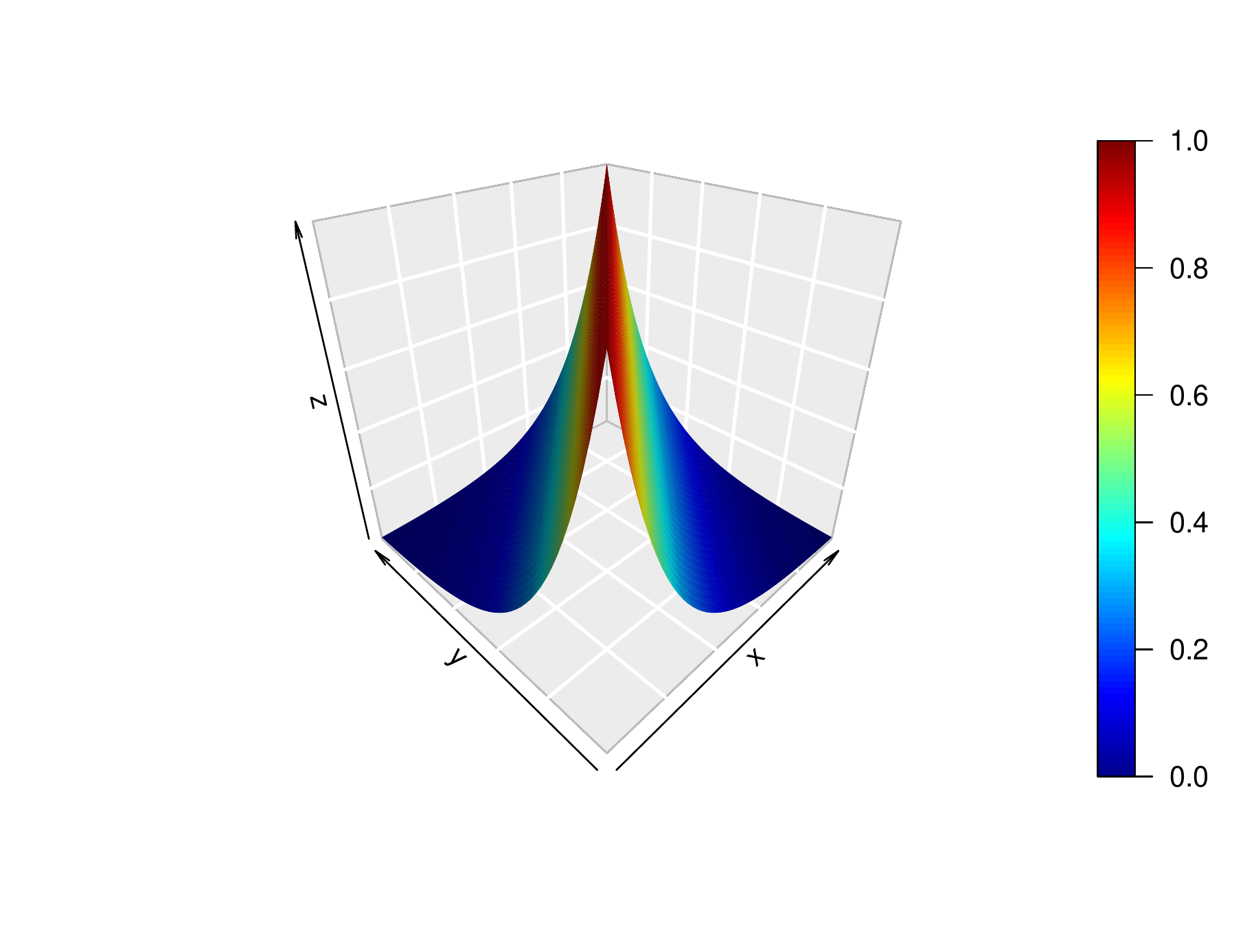} }} \label{surf:matern}%
    \qquad
    \subfloat[\centering ]{{\includegraphics[height=.19\textheight]{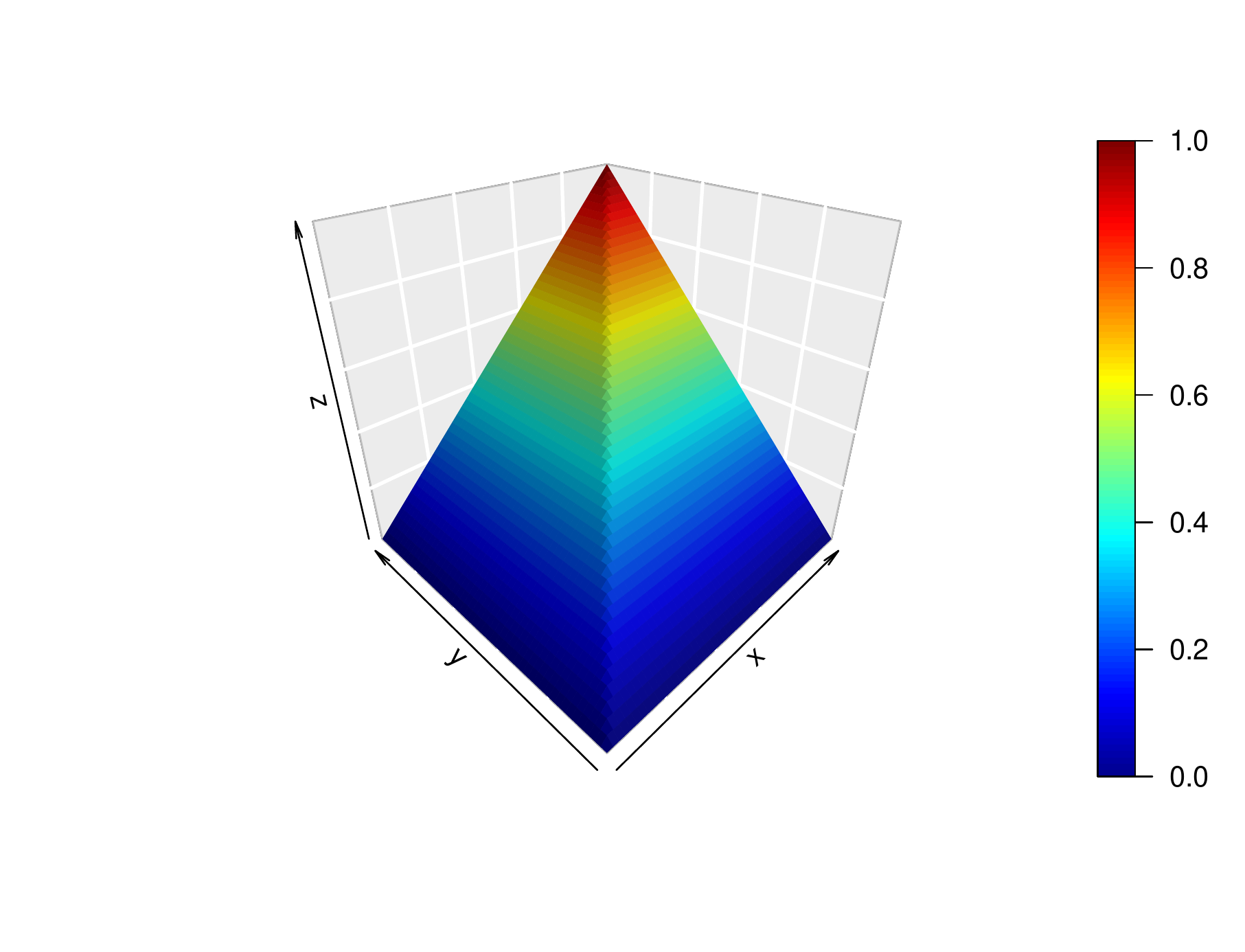} }}%
    \caption{(a): Mat{\`e}rn model: $C(s,t) = \frac{2^{0.5}}{\Gamma(0.5)} \left(\frac{ 1}{5} \right)^{0.5} \vert s-t\vert^{0.5} K_{0.5}\left( \frac{ 1}{5} \vert s-t\vert \right)$, where $\Gamma(\cdot)$ and $K_{0.5}(\cdot)$ are respectively gamma function and modified Bessel function.
    (b): Brownian motion: $C(s,t) = \min(s,t)$. }%
    \label{surface:matern:BM}%
\end{figure}

\subsection{The  {Reflected Triangle} Estimator}

Recall model \eqref{initial model} and write $X_{ij}$ for $X_i(T_{ij})$. For the sake of simplicity, we drop the index of $r_n$ and use $r$ instead, but it is understood that $r$ can vary with $n$. Using the classical decomposition 
\begin{eqnarray*}
 C(s, t)  = \mathbb{E}\left( X(s)X(t) \right) - \mathbb{E}\left( X(s)\right)\mathbb{E}\left( X(t) \right) =: G(s,t) -\mu(s) \mu(t),\;\;\;\;  0 \leq t \leq s \leq 1,
\end{eqnarray*}
we define the estimator
\begin{eqnarray*}
\label{local:lin:C}
 \widehat{C} (s,t) = \widehat{G}(s,t)-\widehat{\mu}(s)\widehat{\mu}(t) ;\;\;\;0 \leq t \leq s \leq 1,
\end{eqnarray*}
 {where $\widehat{\mu}(\cdot)$ and its limiting behaviour can be obtained by employing well established results in literature (e.g. in simulation studies as well as Corollary \ref{dense} below, we choose to use the linear kernel smoothing  method proposed in  \citet{li_uniform_2010}). 
\\
As for the second moment function $G(\cdot, \cdot)$, the main object of interest, we propose } \begin{equation}\label{reflection}\widehat{G}(s,t)=\begin{cases}
\widehat{a}_0,\quad 0\leq s\leq t\leq 1\\
\widehat{G}(t,s),\quad\mbox{otherwise}\\
\end{cases}, \end{equation} 
where $\widehat{a}_0$ is a local quadratic smoother of the restriction of $C$ on $\triangle$, defined via  
\begin{align}\label{local:lin:G}
\left( \widehat{a}_0 , \widehat{a}_1, \widehat{a}_2  \right)& = & \\ \nonumber
&&\underset{a_0,a_1, a_2}{\mathrm{argmin}}\frac{1}{n}\sum_{i=1}^{n}\frac{2}{r(r-1)}\sum_{1\leq k<j \leq r}
\left[\left\{Y_{ik}Y_{ij}-a_0-a_1\left( T_{ij}-s\right)-a_2 \left( T_{ik}-t\right)\right\}^2\right.
\\ \nonumber
&& \hspace{6cm} 
\times \left. K_{H_{G}}\left(\left( T_{ij}-s\right),\left( T_{ik}-t\right)\right)\right], 
\end{align}
for all $0 \leq t \leq s \leq 1$ and $K_{H_G}(\cdot)=K_G(H_G^{1/2}\cdot)$ a kernel with (positive semi-definite) bandwidth matrix $H_G$ dependent on $n$. We highlight that the summation ranges only over $1\leq k < j \leq r$, i.e. only cross product observations belonging to the interior of $\triangle$ contribute to the estimate of $G$ on $\triangle$. The estimate of $G$ on $[0,1]^2\setminus\triangle$ is defined by reflection as in Equation \eqref{reflection}.  
 {
Equivalently, equation \eqref{local:lin:G} assigns weight zero to all cross products $Y_{ij}Y_{ik}$ such that $j\geq k$. Obviously in absence of measurement error we could also include observations on the diagonal and replace the summation set $\{1\leq k<j \leq r\}$ by $\{1\leq k \leq j \leq r\}$. See Figure \ref{sim1:path and cov}.}
\\
   \begin{figure}[t!]%
    \centering
    \subfloat[\centering ]{{\includegraphics[height=.18\textheight]{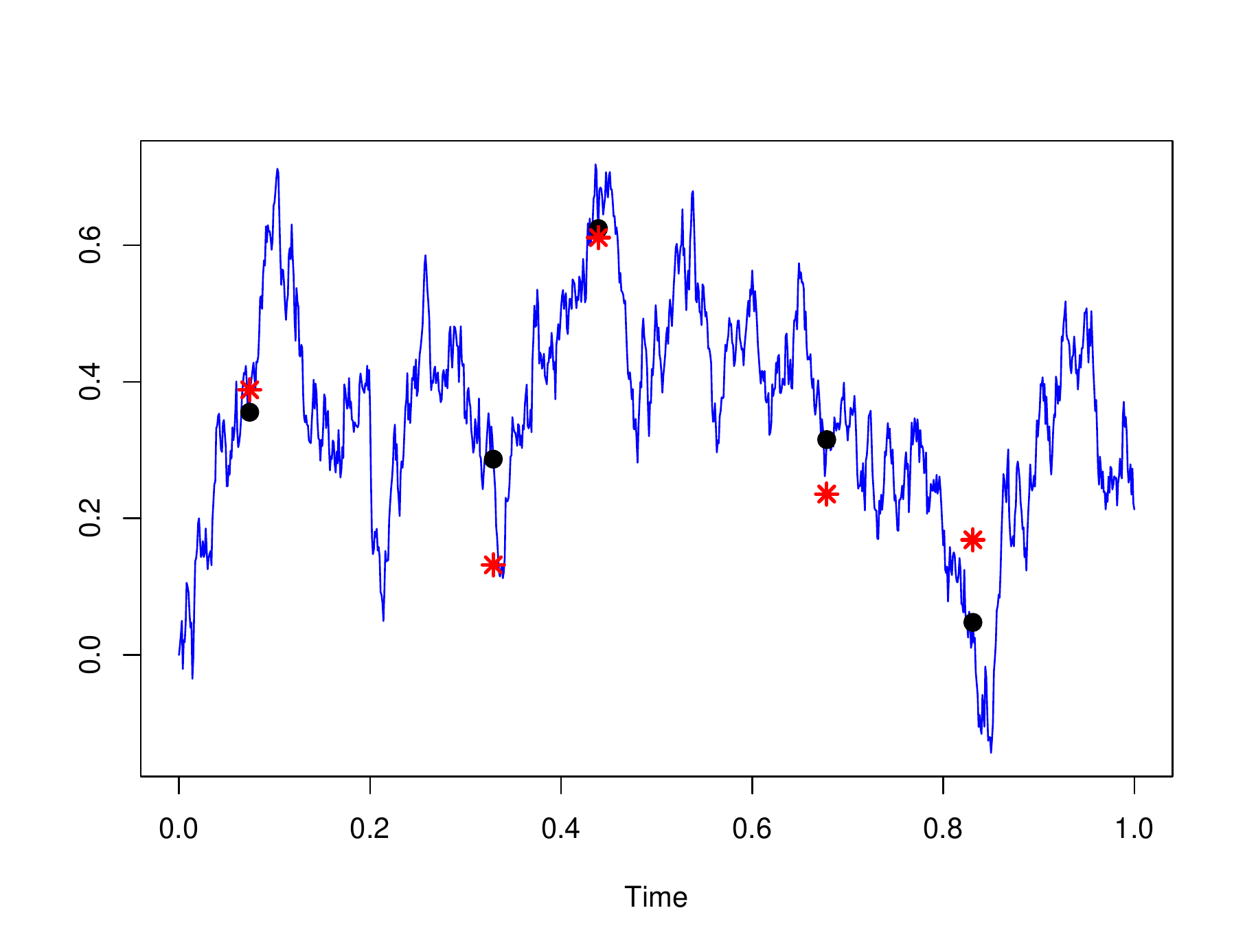} }}%
    \qquad
    \subfloat[\centering ]{{\includegraphics[height=.19\textheight]{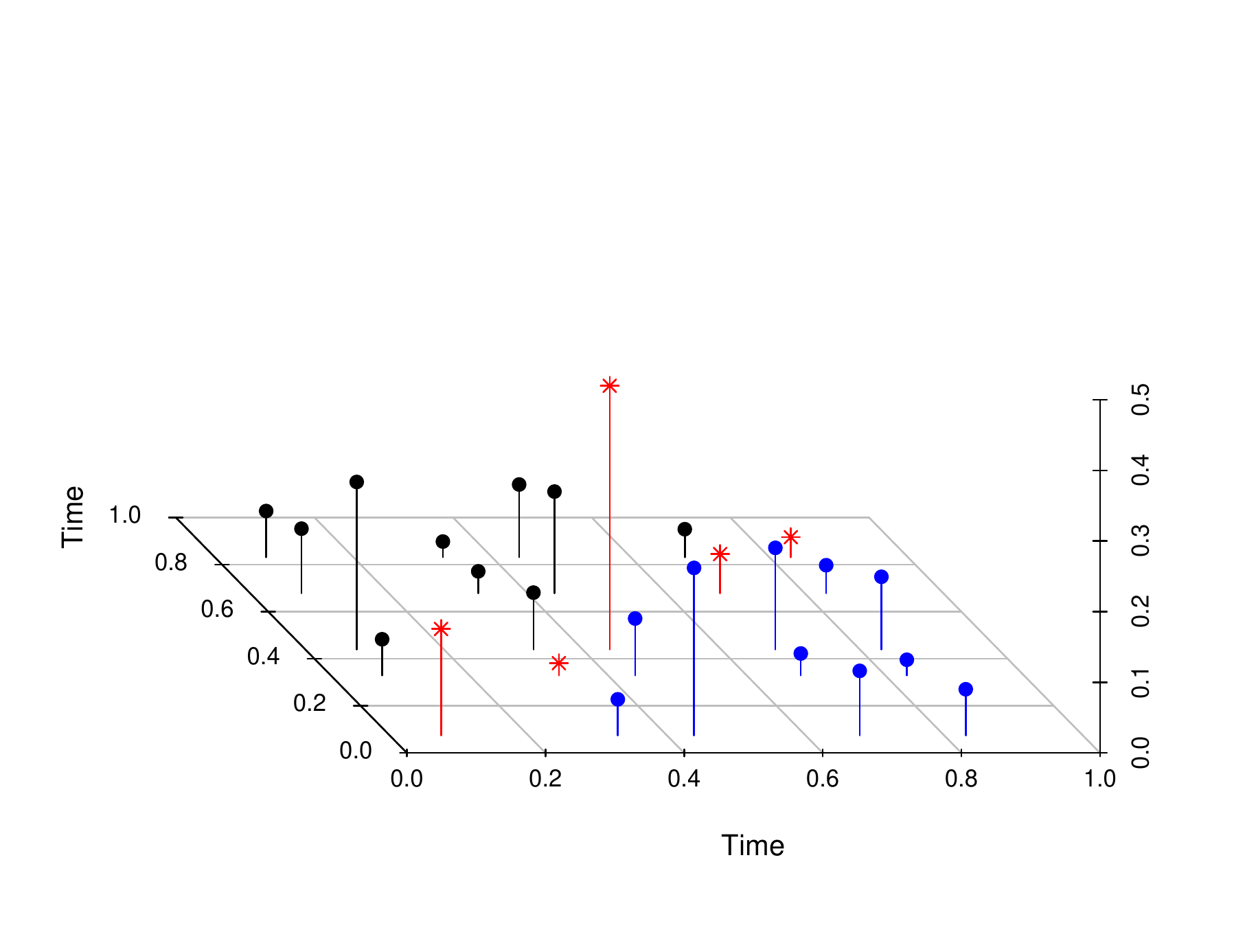} }}%
    \caption{Illustration on a single Brownian motion path: (a) The sample path in blue, five sampled values in black, and the corresponding noise-corrupted observations with red asterisks. (b) The cross products of all noisy observations. Only observations in blue are used for the definition of the triangular estimator. Observations in red are excluded to annhilate the measurement error. And observations in black are excluded from smoothing to avoid smearing the diagonal.}%
    \label{sim1:path and cov}%
\end{figure}

\noindent Despite its unconventional definition, our estimator can still be expressed in closed form, and thus implemented with ease. Define 
\begin{align*}
 \mathbf{T}_{i}^{T} &= \{ (T_{ij},T_{ik}) \}_{1 \leq k,j\leq r},& & \mathbf{Y}_i = \left\{ Y(T_{ij}) Y(T_{ik}) \right\}_{1 \leq k,j\leq r};& i=1,2,\ldots n,\\
 \mathbf{T}_{i,(s,t)} &= \left[ \mathrm{vech}(\mathbf{T}_{i}-(s,t)  \right]; \;\;\; 0 \leq t \leq s \leq 1, & & \mathbb{Y}_i = \left[ \mathrm{vech}(\mathbf{Y}_i)\right];& i=1,2,\ldots n, \\
 \mathbf{T}_{(s,t)} &=  \begin{bmatrix}
          \mathbf{T}_{1,(s,t)} \\
           \mathbf{T}_{2,(s,t)} \\
           \vdots \\
           \mathbf{T}_{n,(s,t)}
         \end{bmatrix},& & \mathbb{Y} = 
         \begin{bmatrix}
         \mathbb{Y}_1\\
         \mathbb{Y}_2\\
         \vdots\\
         \mathbb{Y}_n
         \end{bmatrix},&
         \\
         \mathbb{T}_{(s,t)} &= \left[J,\mathbf{T}_{(s,t)} \right],
\end{align*}
where $\mathrm{vech(M)}$ is a vector obtained by column-wise vectorization  of the matrix $M$ follows by elimination of the diagonal elements as well as upper triangle elements of $M$. The vector $J$,  whose dimension is clear from the context, is a vector with all elements equal to one.
Finally, we define the vector $e_1$, whose dimension is clear from context, to be the vector with leading element equals to one and zeroes elsewhere, and 
\begin{eqnarray*}
\mathbf{W}_{(s,t)} &=& \mathrm{diag} \left\{  K_{H_{G}} \left(\mathbf{T}_{(s,t)} \right) \right\},
\end{eqnarray*}
where  $K_{H_{G}} \left(\mathbf{T}_{(s,t)} \right)$
should be understood element-wise. \\

\noindent In this notation, the local least square problem \eqref{local:lin:G}  
admits the following closed-form solution
\begin{eqnarray}\label{a_0}
 \widehat{a}_0 = \widehat{G}(s,t) =  e_{1}^{T}\left(\mathbb{T}_{(s,t)}^{T}\mathbf{W}_{(s,t)}\mathbb{T}_{(s,t)}\right)^{-1}\mathbb{T}_{(s,t)}^{T}\mathbf{W}_{(s,t)} \mathbb{Y}.
\end{eqnarray}
 See \citet{ruppert_multivariate_1994} for more details. In the remainder of the paper, we assume  $H_G$ to be diagonal, $H_G  = \mathrm{diag}\left( h^2_G, h^2_G\right)$. We set the $2$-variate kernel $K_G(\cdot,\cdot) $ as the product kernel based on symmetric univariate integrable kernel $W(\cdot)$ i.e. $K_G(\cdot,\cdot)  =W(\cdot)W(\cdot)  $. 
 The univariate kernel $W(\cdot)$ depends in general on $n$ and can be defined in one of the forms $\left\{ \mathbf{W}_n \right\}$ or  $\left\{ \mathbb{W}_n \right\}$ below:
\begin{eqnarray}\label{kernel1} \mathbf{W}_n(u) =  \mathrm{exp}\left( -\frac{u}{\sigma_n} \right),
\end{eqnarray}
\begin{eqnarray}\label{kernel2}
\mathbb{W}_n (u)= 
\left\{
\begin{array}{ll}
\mathbf{W}_1 (u)    &\mathrm{if} \;\vert u \vert < 1,  \\
\mathbf{W}_n (u)   & \mathrm{if} \;\vert u \vert \geq 1,
\end{array}
\right.
\end{eqnarray}
where  $\left\{ \sigma_n \right\}_n$ is a sequences of positive numbers tending to zero sufficiently fast such that $W_n(1^+)  \lesssim O\left(h_G^4  \right)$.  Section \ref{Appendix} explains this particular choice of kernels and  $W(\cdot)$, and discusses other choices.

\subsection{Asymptotic Theory} \label{Asymptotic Theory}
We will now establish the almost sure uniform convergence rates of the estimators $\widehat{G}(\cdot, \cdot) $  
appearing in \eqref{a_0}. 
Our conditions can be stated as follows:
\begin{itemize}
    \item[\namedlabel{bound:density}{C(0)}] There exists some positive number $M_T$ for which we have $0 < \mathbb{P}\left(T_{ij} \in [a,b] \right) \leq M_T (b-a)$, for all $i,j$ and $0 <a < b< 1$.
    
    \item[\namedlabel{sup|.|^a}{C(1)}] $\mathbb{E}\vert U_{11}      \vert^{\alpha}< \infty$, and  $\underset{0\leq s \leq 1}{\mathrm{sup}} \mathbb{E} \vert X_{1}(s)\vert^{\alpha} < \infty $.
    
    
    \item[\namedlabel{2-diff:C}{C(2)}] The second order partial derivatives of $C$  on the (lower)  triangle $\bigtriangleup=\left\{(s,t)\in [0,1]^2| 0 \leq t \leq s \leq 1  \right\}$  exist and are bounded \textit{i.e.} $\left.C\right|_\triangle \in \mathcal{C}^2(\bigtriangleup)$;  differentiability on boundary points is defined via one-sided limits.
    
    
    

\end{itemize}

Assumptions \ref{bound:density} and  \ref{sup|.|^a}  are standard. Assumption \ref{2-diff:C} relaxes the usual smoothness assumption on the covariance, to accommodate non-smooth paths. Our main results on the almost sure convergence rates of our estimator is now as follows:
\begin{theorem}\label{Thm:Ghat-G}
Under Conditions \ref{bound:density}--\ref{2-diff:C}, for   $\alpha \in (4,\infty)$, it holds with probability 1 that
\begin{align}\label{bound:Ghat-G}
 \underset{0\leq t \leq s \leq 1}{\mathrm{sup}}\left|\widehat{G}(s, t) - G(s, t)\right|
 &=&  O\left[h_{G}^{-4} \frac{\mathrm{log}n}{n}\left( h_G^4 + \frac{h_G^3}{r}+ \frac{h_G^2}{r^2} \right) \right]^{1/2}+O\left( h^2_G  \right).
\end{align}
\end{theorem}
The following corollary states sufficient conditions such that, under a dense sampling scheme,  the right hand side of \eqref{bound:Ghat-G}  will be dominated by $(\log n/n  )^{1/2}$.  {This  similarly happens for the discrepancy 
$\widehat{\mu}(\cdot)$ and $\mu(\cdot)$, where $\widehat{\mu}$ is obtained based on an application of locally weighted least square regression method with the same univariate kernel $W(\cdot)$ as above and bandwidth parameter $h_{\mu}$, see Theorem 3.1 of \citet{li_uniform_2010}. }
 It thus follows that, under dense observation, the rate of our covariance estimator matches the optimal rate under complete observation (see Corollary 4.1 of \citet{bosq_linear_2000}), bridging the gap with ``fully observed theory" that does not assume smoothness (as described in the introduction). 
\begin{corollary}[Rates Under Dense Observation]\label{dense}
Assume dense sampling scheme, i.e. that $r_n \geq M_n$, for some increasing sequence $M_n\uparrow\infty$,   and $M_n^{-1} \lesssim h_{\mu} , h_{G} \lesssim \left( \frac{\mathrm{log}n}{n} \right)^{1/4} $, and  Conditions \ref{bound:density}--\ref{2-diff:C}, for   $\alpha \in (4,\infty)$. Then it holds with probability 1 that
$$\underset{0\leq t \leq s \leq 1}{\mathrm{sup}} \vert C(s,t) - \widehat{C}(s,t) \vert  = O  \left( \frac{\mathrm{log}n}{n} \right)^{1/2}.$$
In particular we have, with probability 1, that $\Vert C - \widehat{C}  \Vert_{HS}^2  =\int_0^1\int_0^1 (C(s,t) - \widehat{C} (s,t))^2 dsdt  = O  \left( \frac{\mathrm{log}n}{n} \right)$.
\end{corollary}
\section{Simulations}\label{simulation}

To illustrate the finite-sample performance of our method, we run a small simulation study on two examples, one involving a nowhere differentiable process, and one involving a $\mathcal{C}^2$ process. In either case, we apply our ``triangle" method as well as the classical ``square" smoothing methods (\citet{yao_functional_2005} and \citet{li_uniform_2010}, for example), and compare the results. Though our method is applicable in both examples (since it rests on weaker assumptions), the classical covariance smoothing method only applies to the smooth example.  Naturally, we expect our method to perform better in the rough example, and worse in the smooth example (since it does not take full advantage of the smoothness). The interesting point we observe, however, is that our method performs considerably better in the rough case, and comparably well in the $\mathcal{C}^2$ case.   {See Figures \ref{sim1:surfaces} and \ref{sim2:surfaces} which depict the true underlying covariance functions  as well as the resulting estimated functions based on one of the simulation runs applying each of the methods discussed above.}\\

\noindent \textbf{Example 1: } First, we consider the example of a standard Brownian motion, with covariance of the form $C(s,t) = \mathrm{min}\{s,t\}$, $0 \leq s,t \leq 1$. This is of class $\mathcal{C}^{\infty}$ on $\bigtriangleup$, but fails to be $\mathcal{C}^1$  on the unit square. We take $n=100$, and $r\in\{5,20,50\}$, and run $100$ replications in each case.  The results are summarised in Figures \ref{sim1:box:HS_dist}-\ref{sim1:eigvecs}; to allow for a compact presentation, each figure uses blue to demonstrate the target value; gray or black are used to denote the values corresponding to the triangular method; and red is assigned to the square method. The proposed method considerably outperforms the existing classical approach in estimating the covariance kernel, as well as its spectrum. Upon closer inspection we saw that the effect was particularly pronounced near the diagonal, as is reasonable to expect. This is reflected in the estimated eigenvalues: the triangular method performs better in estimating the leading eigenvalues, and substantially better for eigenvalues of larger order, see Figures \ref{sim1:box:eigvals1-2}-\ref{sim1:box:eigvals10-20}. Figure \ref{sim1:box:eigvecs} similarly shows that the triangular method performs well in estimating all the leading eigenfunctions, even for as few as $r=5$ sampled points per curve. By contrast the square method generally encounters considerable difficulties with eigenfunctions other than the leading one, and require considerably higher values of $r$ in order to capture higher order eigenfunctions.

\begin{figure}[t!]%
    \centering
    \subfloat[\centering ]{{\includegraphics[height=.11\textheight]{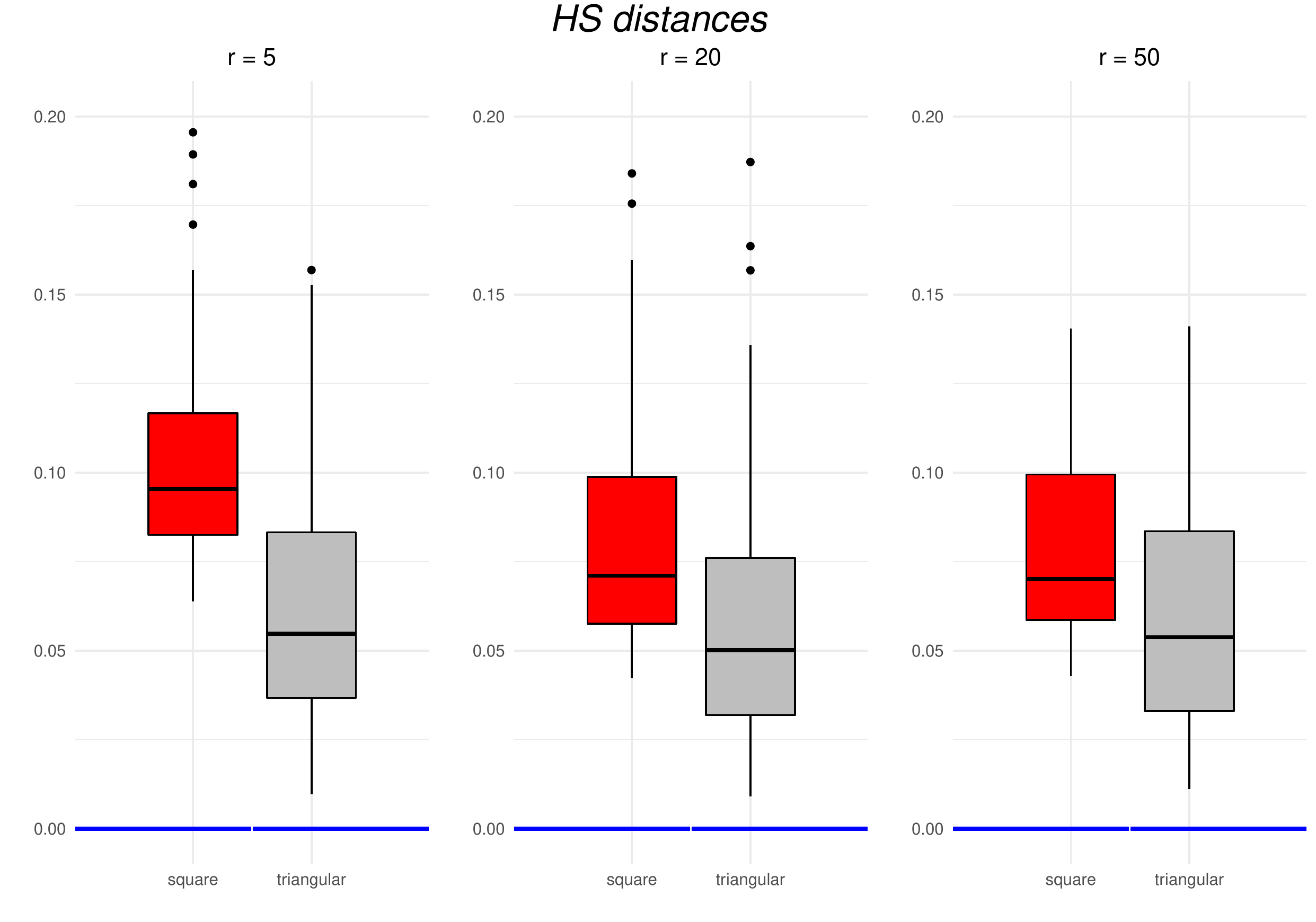}  }
    \label{sim1:box:HS_dist}}%
    \qquad
    \subfloat[\centering ]{{\includegraphics[height=.11\textheight]{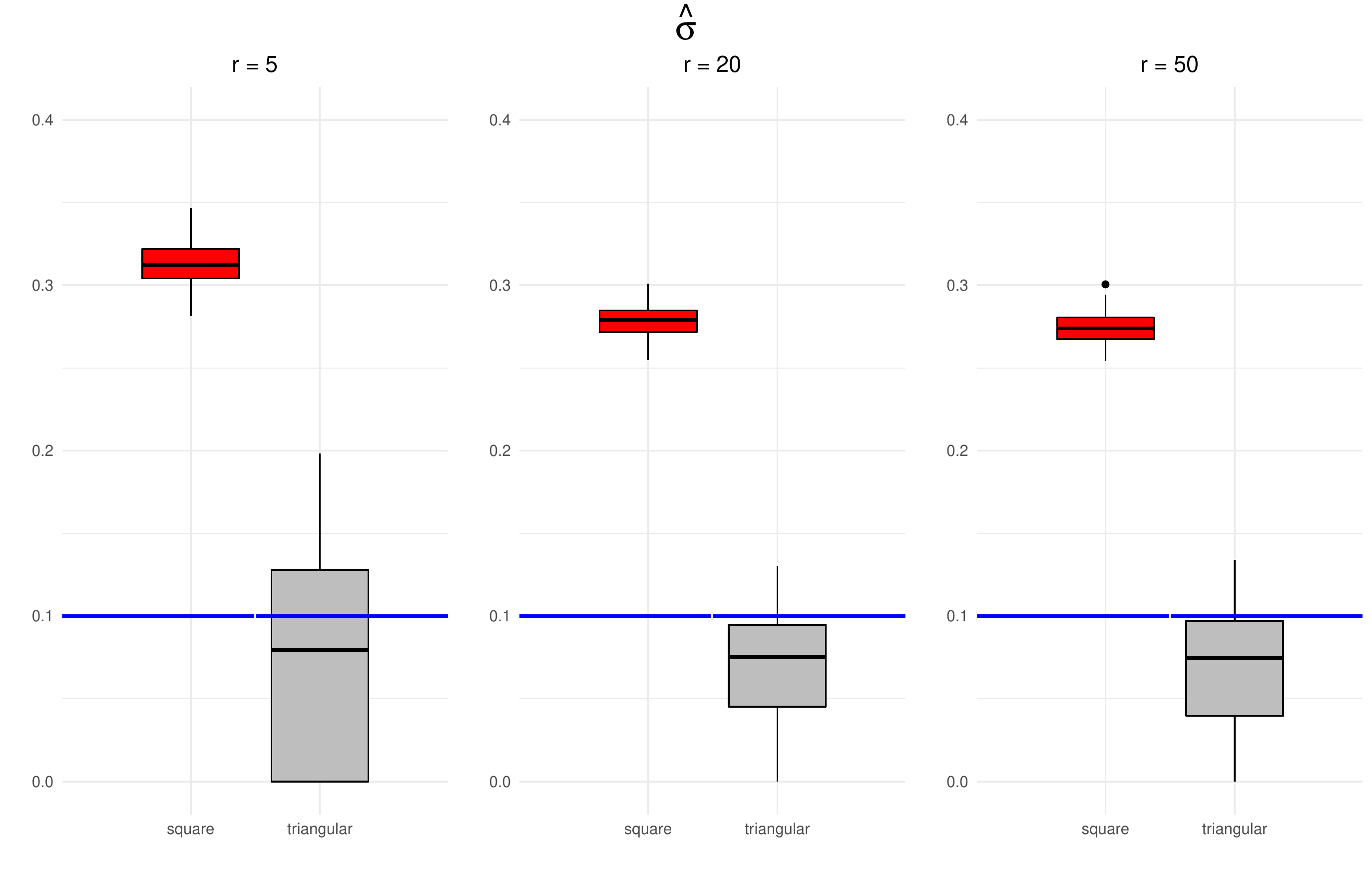} } \label{sim1:box:std_err_hat}}%
    \quad
    \subfloat[\centering ]{{\includegraphics[height=.11\textheight]{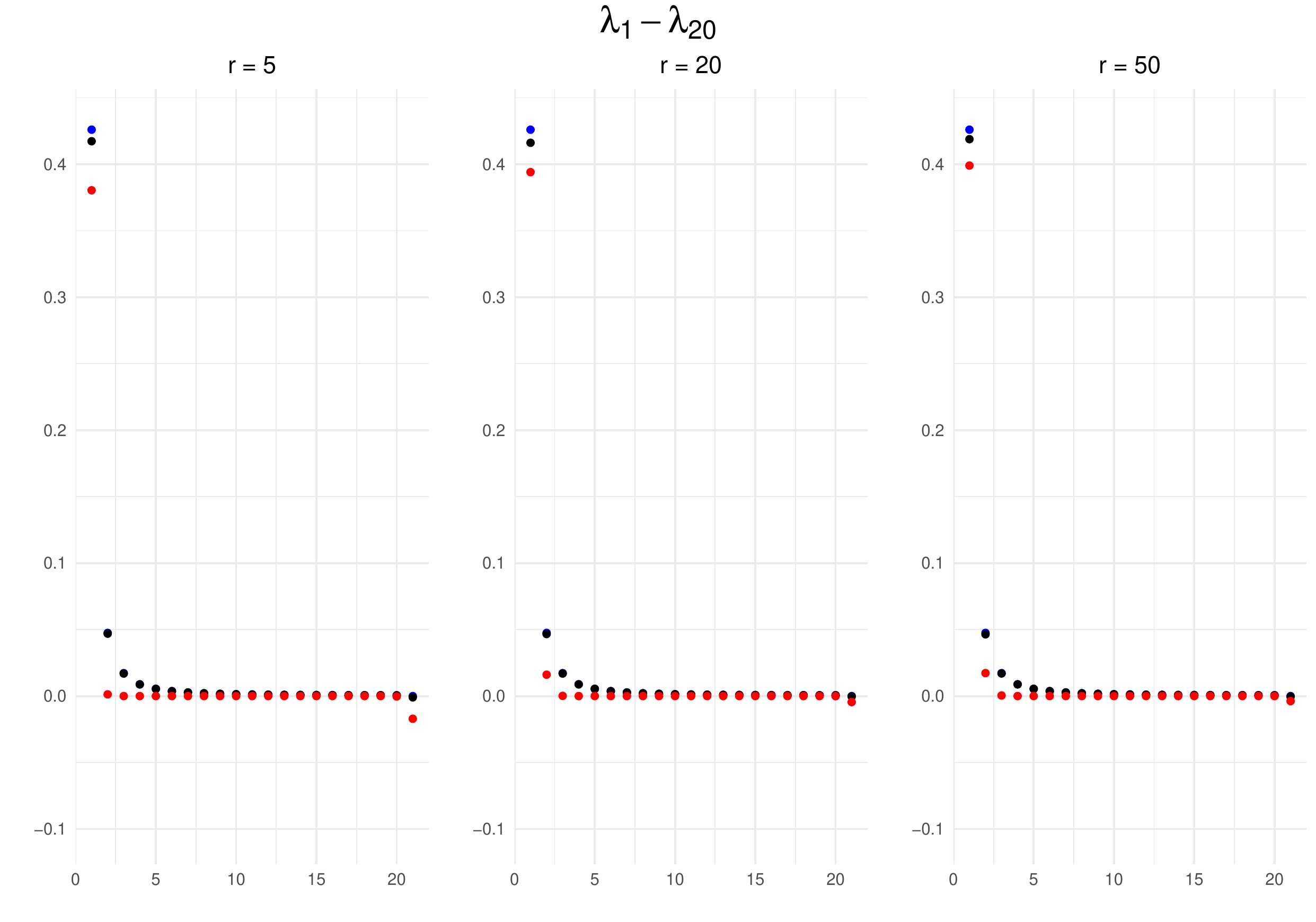} }\label{sim1:eigvals}}%
    \caption{Example 1: (a) Box plots for $\Vert C - \widehat{C} \Vert_{HS}$. (b) Box plots for $\widehat{\sigma}$. (c) The first 20 eigenvalues, $\lambda_1$ to $ \lambda_{20}$,  and the average of their estimates over the simulation runs.}%
    \label{sim2:surfaces}%
\end{figure}

\begin{figure}[t!]%
    \centering
    \subfloat[\centering ]{{\includegraphics[height=.11\textheight]{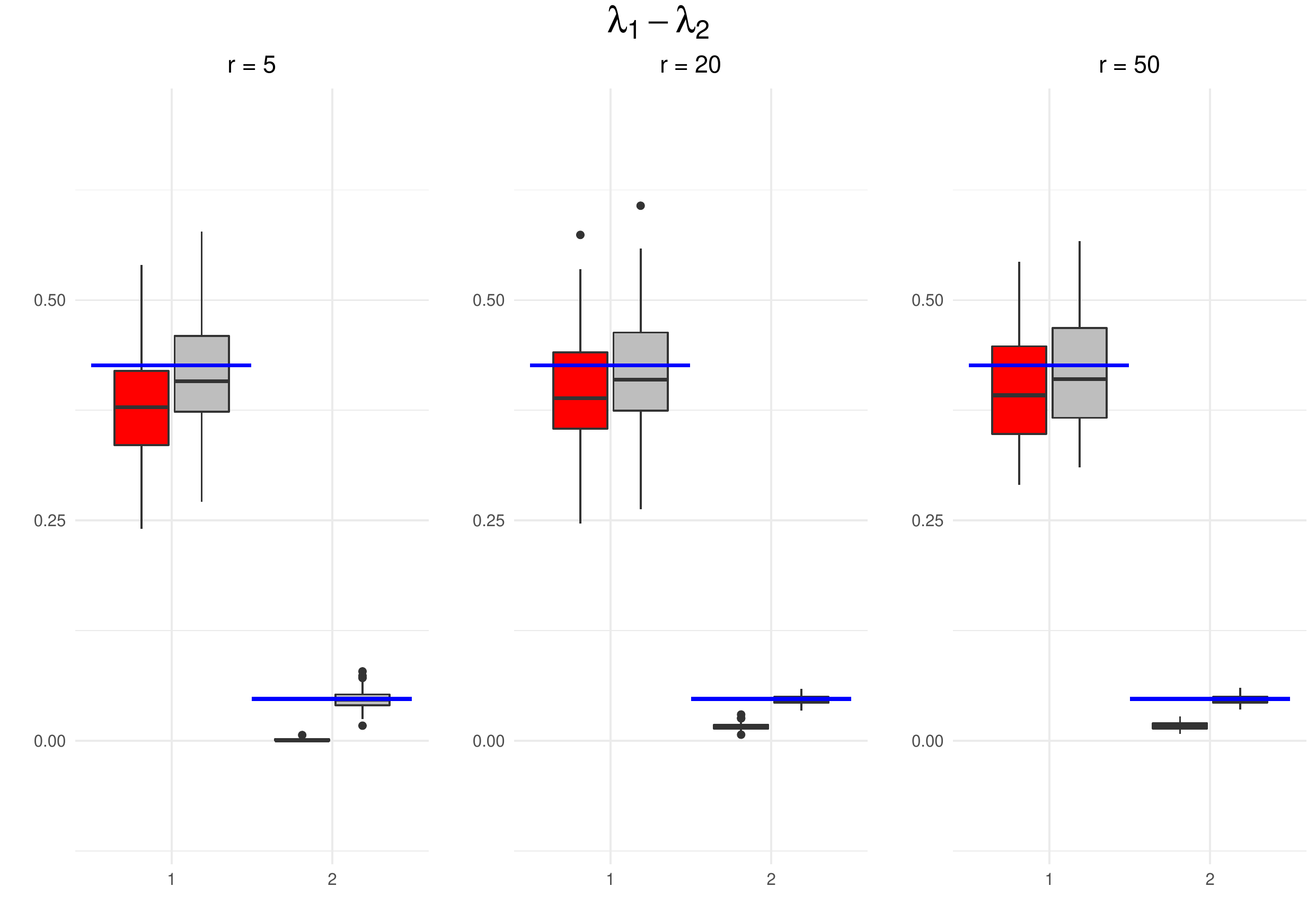} }
   \label{sim1:box:eigvals1-2}}%
    \qquad
    \subfloat[\centering ]{{\includegraphics[height=.11\textheight]{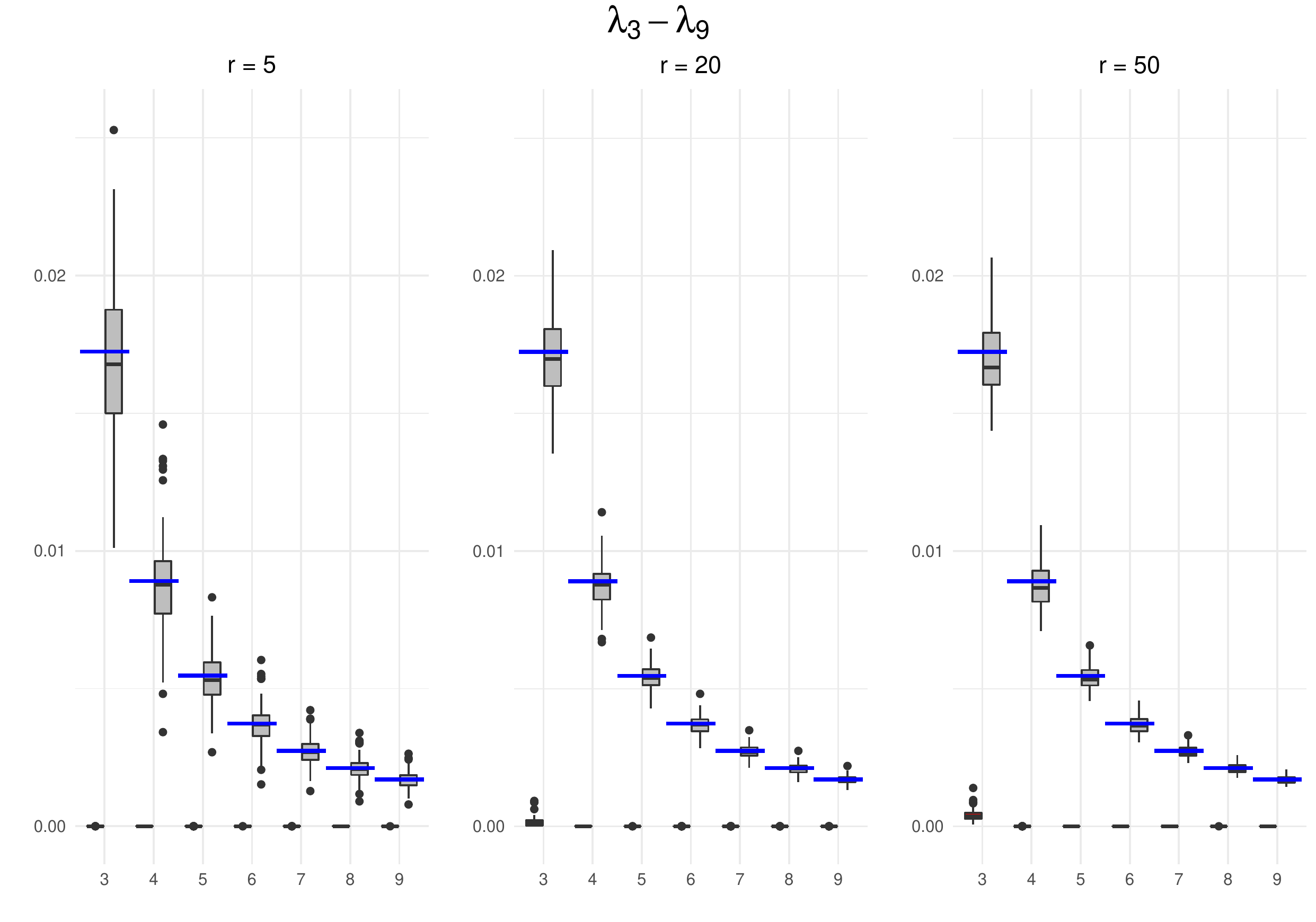} } \label{sim1:box:eigvals3-9}}%
    \quad
    \subfloat[\centering ]{{\includegraphics[height=.11\textheight]{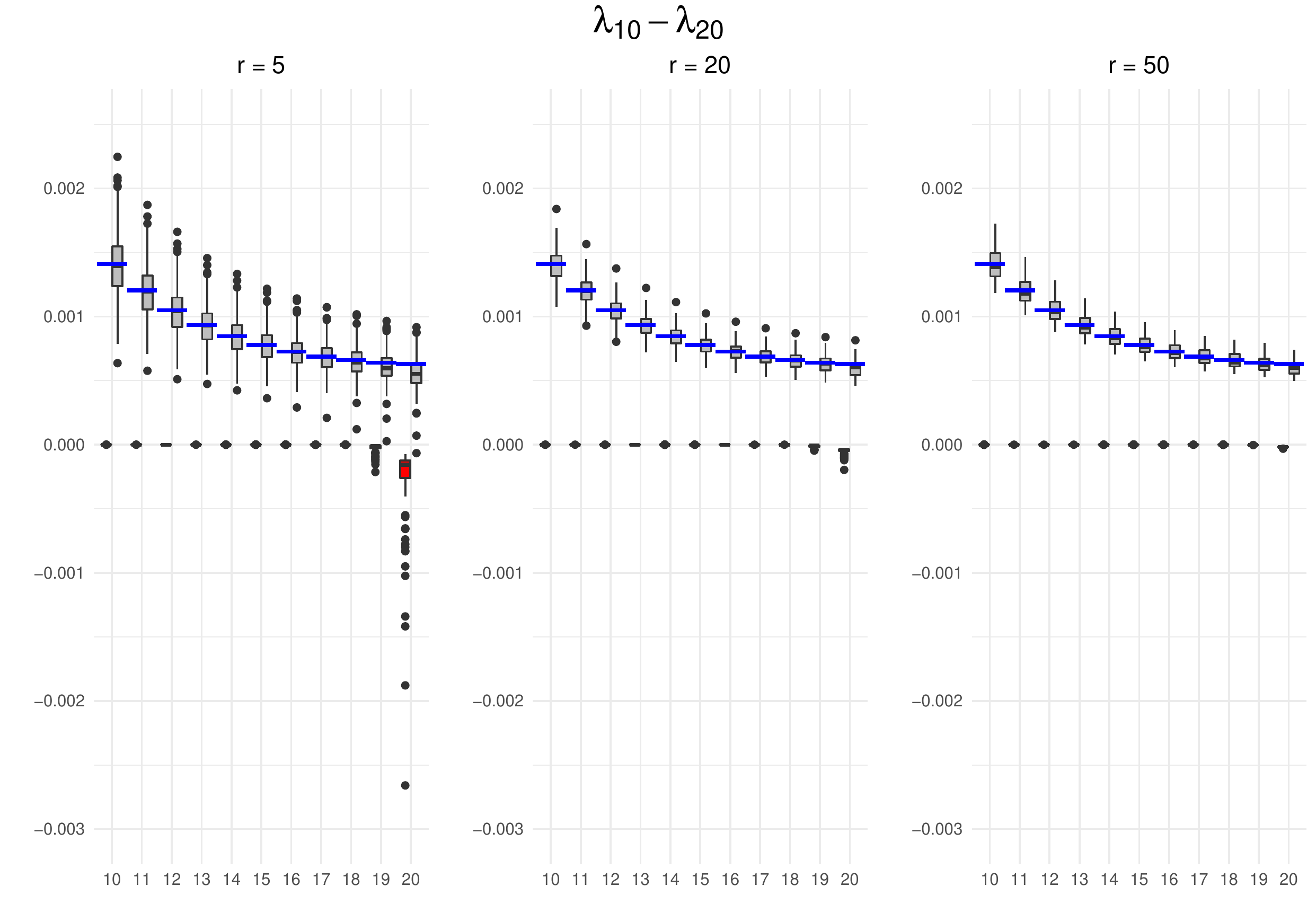} } \label{sim1:box:eigvals10-20}}%
    \caption{Example1: (a)  Box plots for $\widehat{\lambda}_1$ and $\widehat{\lambda}_2$. (b)  Box plots for $\widehat{\lambda}_3$ to $\widehat{\lambda}_9$.(c)  Box plots for $\widehat{\lambda}_{10}$ to $\widehat{\lambda}_{20}$.}%
    \label{sim2:surfaces}%
\end{figure}
\begin{figure}[t!]%
    \centering
    \subfloat[\centering ]{{\includegraphics[height=.12\textheight]{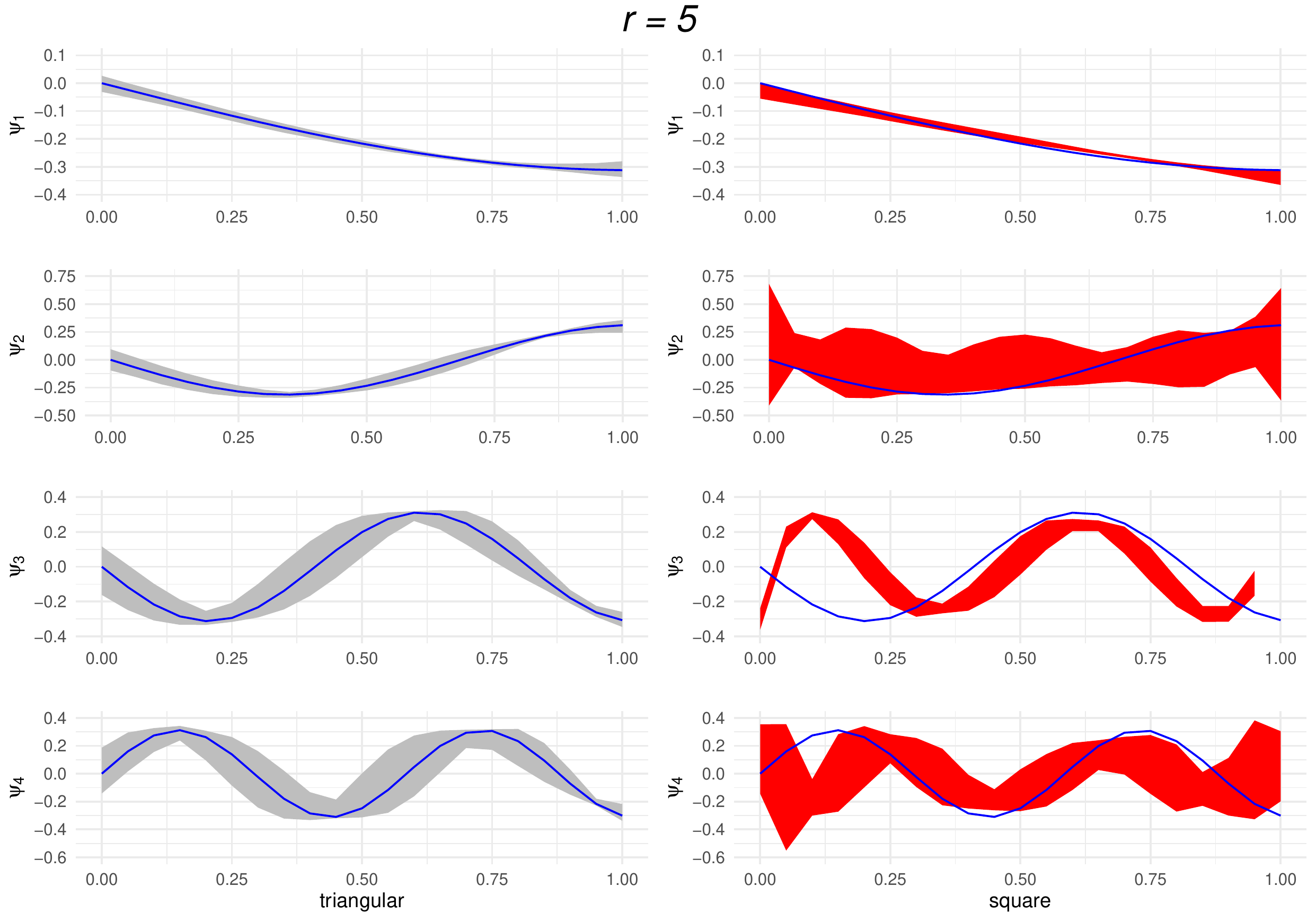} }
   }%
    \qquad
    \subfloat[\centering ]{{\includegraphics[height=.12\textheight]{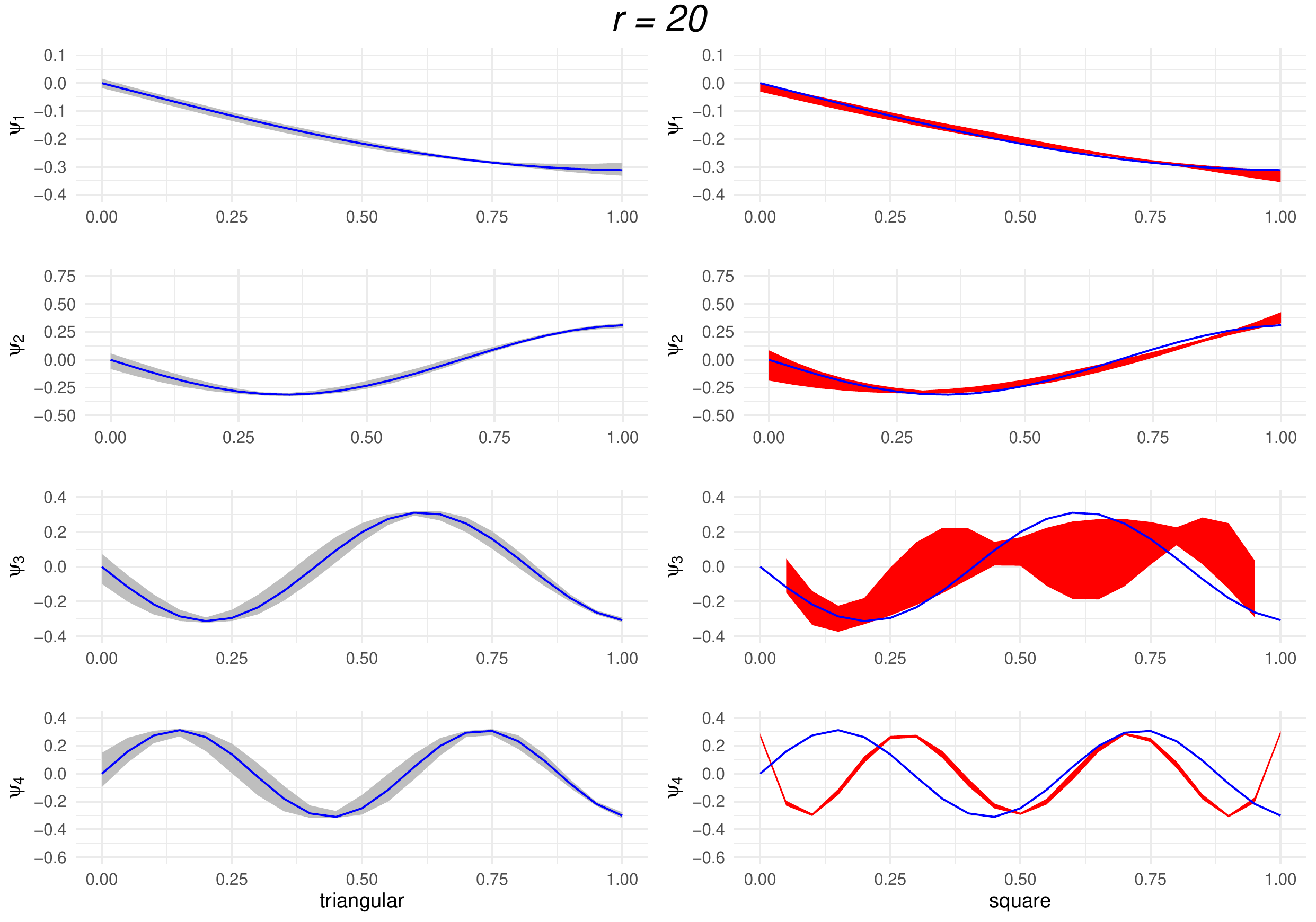} } 
    }%
    \quad
    \subfloat[\centering ]{{\includegraphics[height=.12\textheight]{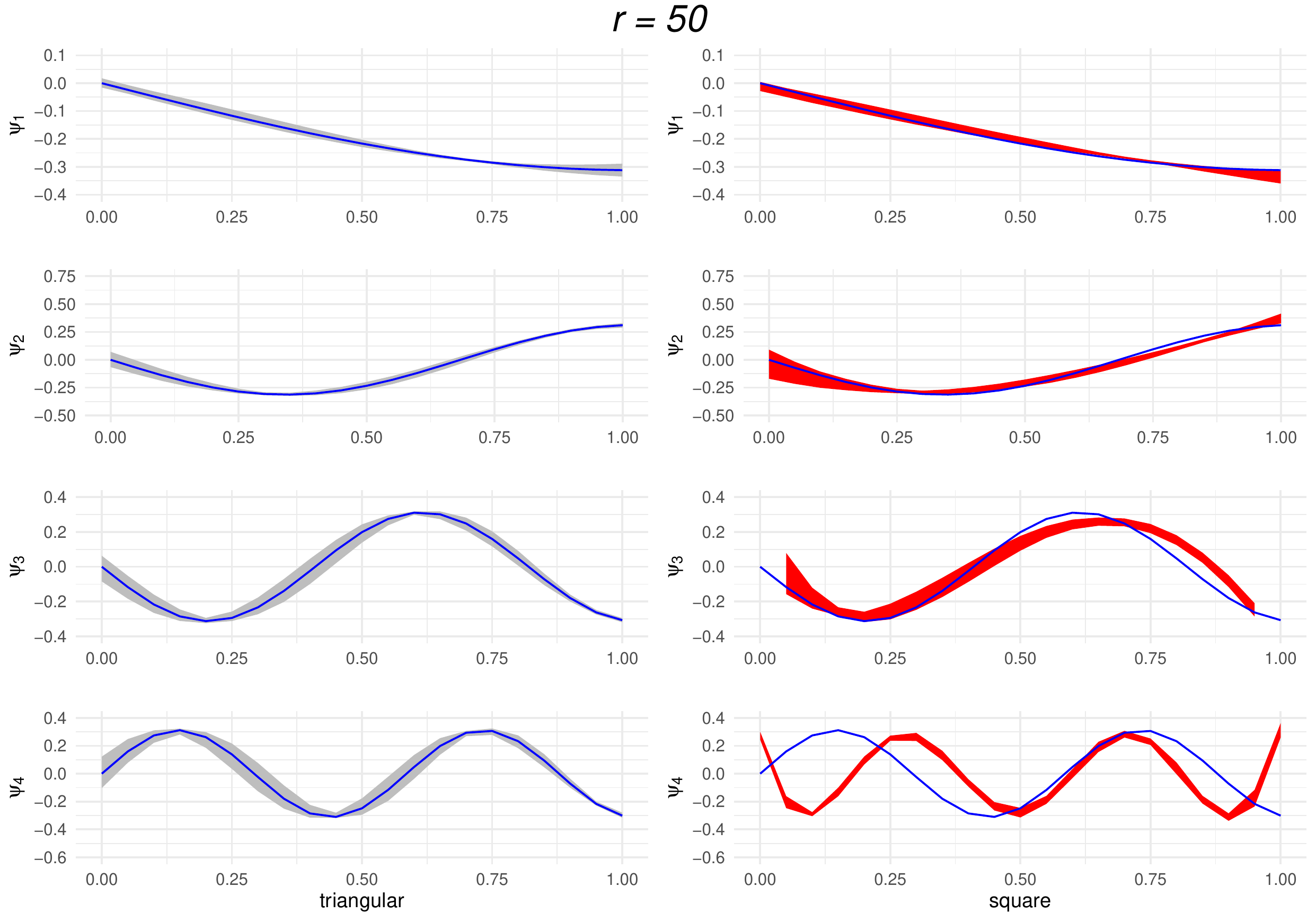}} 
    \label{sim1:eigvecs}
    }%
    \caption{Example 1, The first 4 eigenfunctions, $\Psi_1$ to $\Psi_4$, and point-wise $95\%$ confidence  envelopes based on the simulations. (a) $r=5$. (b) $r=20$. (c) $r=50$. }%
    \label{sim1:box:eigvecs}
\end{figure}

   \begin{figure}[h!]%
    \centering
    \subfloat[\centering ]{{\includegraphics[height=.13\textheight]{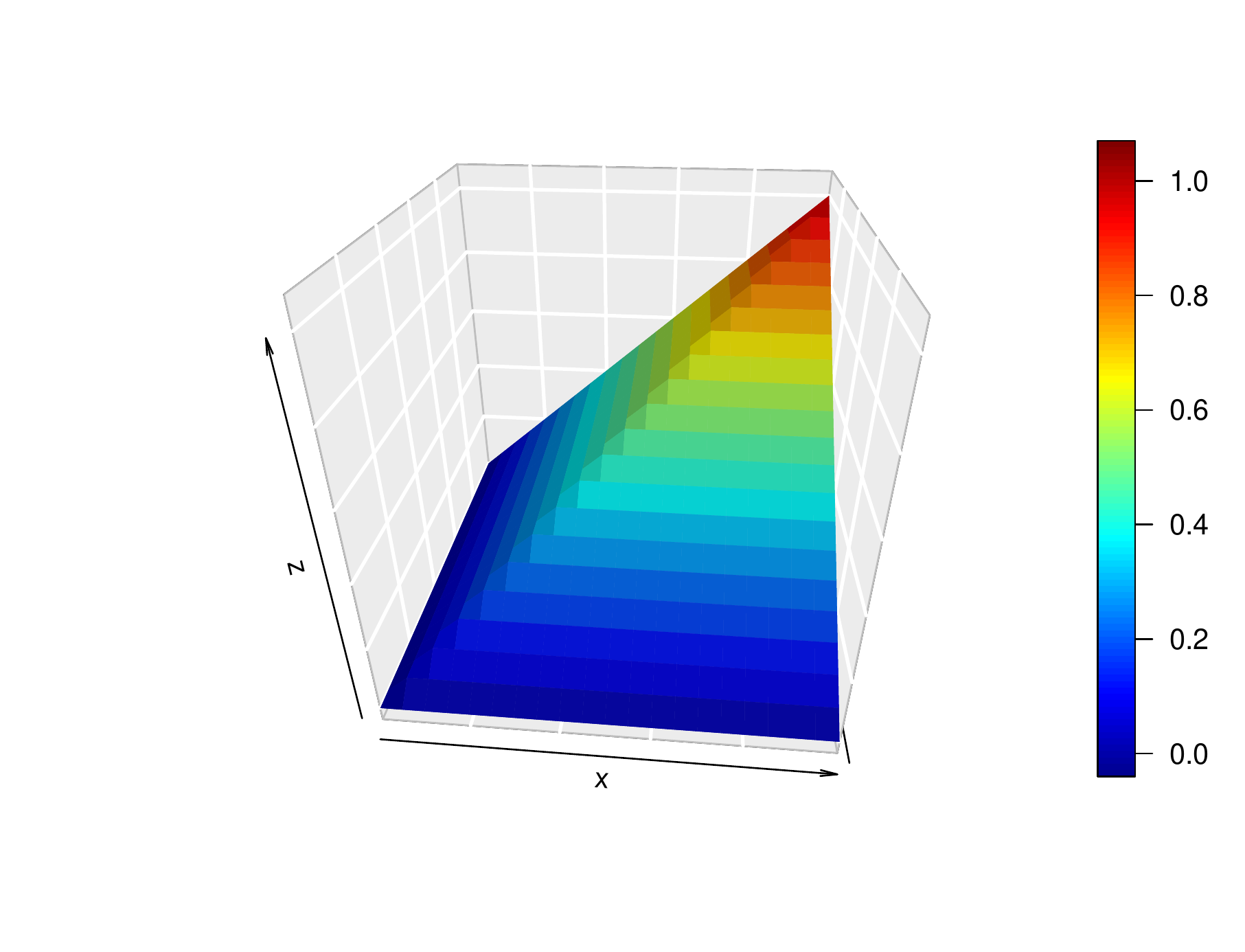} }}%
    \qquad
    \subfloat[\centering ]{{\includegraphics[height=.13\textheight]{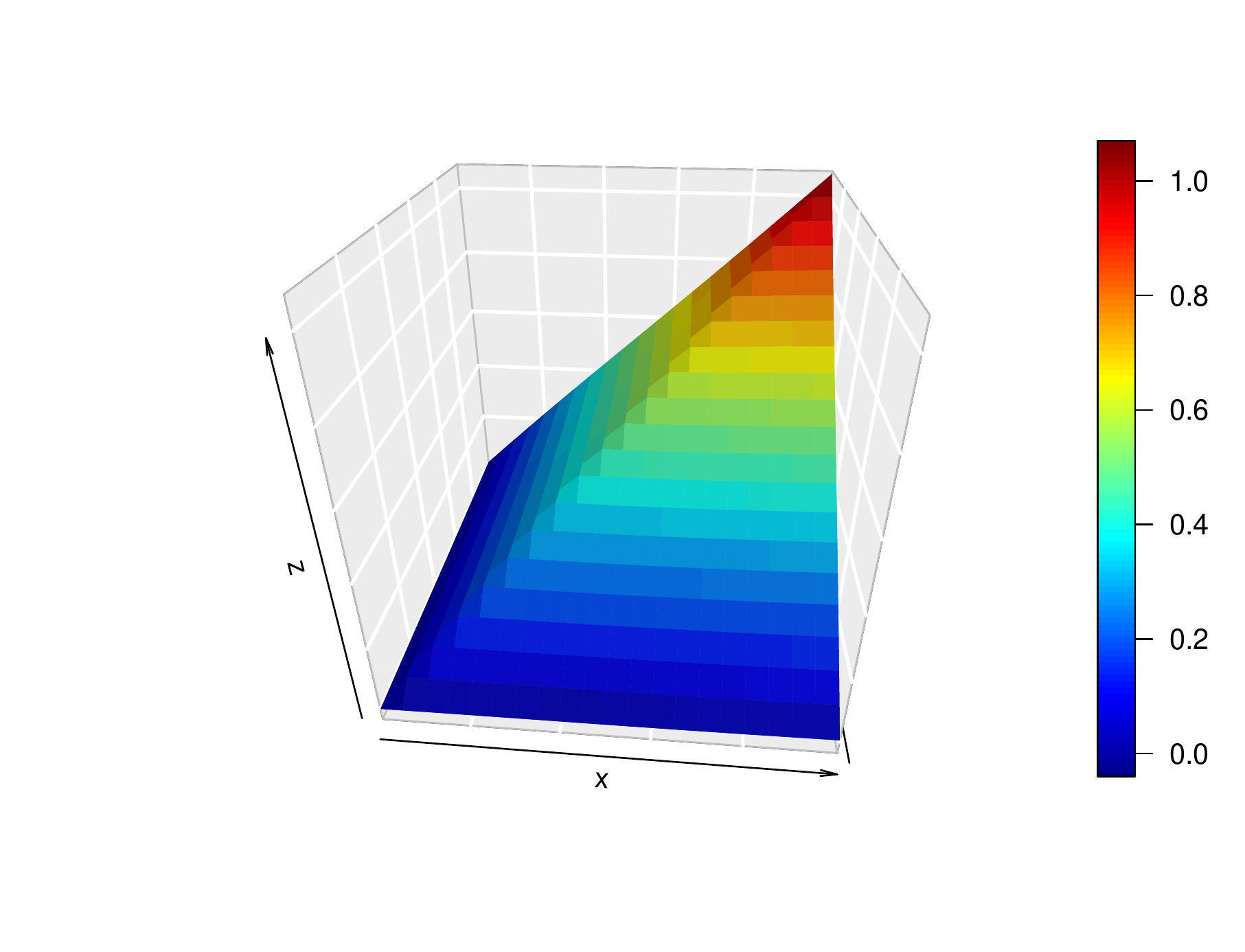} }}%
    \quad
    \subfloat[\centering ]{{\includegraphics[height=.13\textheight]{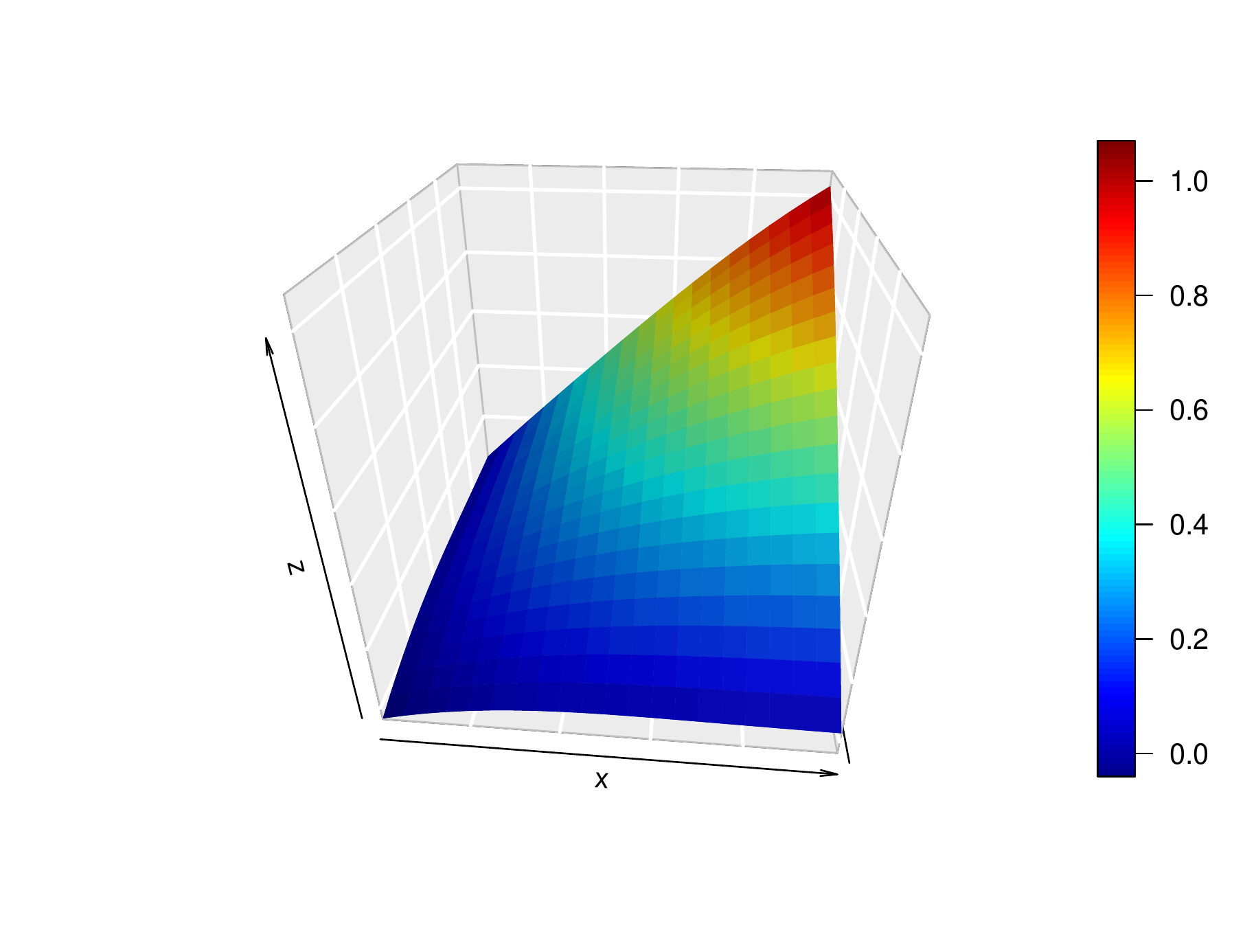} }}%
    \caption{(a): Covariance function of Brownian motion, $C(s,t) = \min(s,t)$. (b): $\widehat{C}$, based on one of the simulation runs, applying the proposed  {reflected triangle}  approach. (c): $\widehat{C}$, based on  the same simulation run as part (b), applying the squared domain approach.}%
    \label{sim1:surfaces}%
\end{figure}

\noindent \textbf{Example 2: } Next, we consider a process $X(s) = \mu(s) + Z_1 g_1(s) + Z_2 g_2(s) + Z_3 g_3(s)$, where $\mu(s) = 5(s-0.6)^2$, $Z_j\stackrel{iid}{\sim}N(0, 0.2)$,  $g_1(s) = 1$ is a constant function, $g_2(s) = \mathrm{sin}(2\pi s)$, and $g_3$ is chosen to be (a multiple of) four-fold convolution of the constant function $1$ over $\left[ 0 , \frac{1}{4} \right]$ \textit{i.e.}
\begin{eqnarray}
\hspace{1cm}
g_3(s)= 
\left\{
\begin{array}{ll}
 s^3    &\mathrm{for} \;0 \leq s < \frac{1}{4},  \\
  s^3 - 4  \left(s - \frac{1}{4}\right)^3     & \mathrm{for} \;\frac{1}{4} \leq s < \frac{2}{4}, \\
   s^3 - 4  \left(s - \frac{1}{4}\right)^3 + 6\left(s - \frac{2}{4}\right)^3   & \mathrm{for} \;\frac{2}{4} \leq s < \frac{3}{4}, \\
   s^3 - 4  \left(s - \frac{1}{4}\right)^3 + 6\left(s - \frac{2}{4}\right)^3 - 4  \left(s - \frac{3}{4}\right)^3   & \mathrm{for} \;\frac{3}{4} \leq s \leq 1. 
\end{array}
\right.
\end{eqnarray}
The function $g_3(s)$ is constructed to have a shape similar to that of $\mathrm{cos}(2 \pi s)$. 
Consequently, this simulation scenario yields a setting similar to that of Simulation 1 in \citet{li_uniform_2010} with the difference that, clearly, the covariance kernel here is of class $\mathcal{C}^2$ but not $\mathcal{C}^3$ on $\bigtriangleup$ (or the unit square). The point of this modification is to see the net effect of going from a $\mathcal{C}^0$ process to a $\mathcal{C}^2$ process (rather than from $\mathcal{C}^0$ to $\mathcal{C}^{\infty}$, if we had picked $\mathrm{cos}(2 \pi s)$ instead of $g_3(s)$).  We take $n=200$, $r\in\{5,20,50\}$, and run $100$ replications in each case.  The simulation results are summarised in Figures \ref{sim2:box:HS_dist}-\ref{sim2:eigvecs}. We use the same colour coding as before: blue for the target, grey/black for the triangular method, and red for the square method. Interestingly, the triangular method performs slightly better than the square method when it comes to the covariance as  a whole (Figure \ref{sim2:box:HS_dist}). The two approaches perform very similarly in estimating the leading eigenfunctions (Figure \ref{sim2:box:eigvecs}). When it comes to the eigenvalues, we see an interesting pattern: both methods perform comparably in the the non zero eigenvalues (up to order 3).  The square method does a clearly better job at capturing the fact that the remaining eigenvalues are zero (except when going to higher orders where both methods start having trouble). See Figures \ref{sim2:box:eigvals1-2}-\ref{sim2:box:eigvals10-20}. 

\begin{figure}[t!]
   \centering
    \subfloat[\centering ]{{\includegraphics[height=.11\textheight]{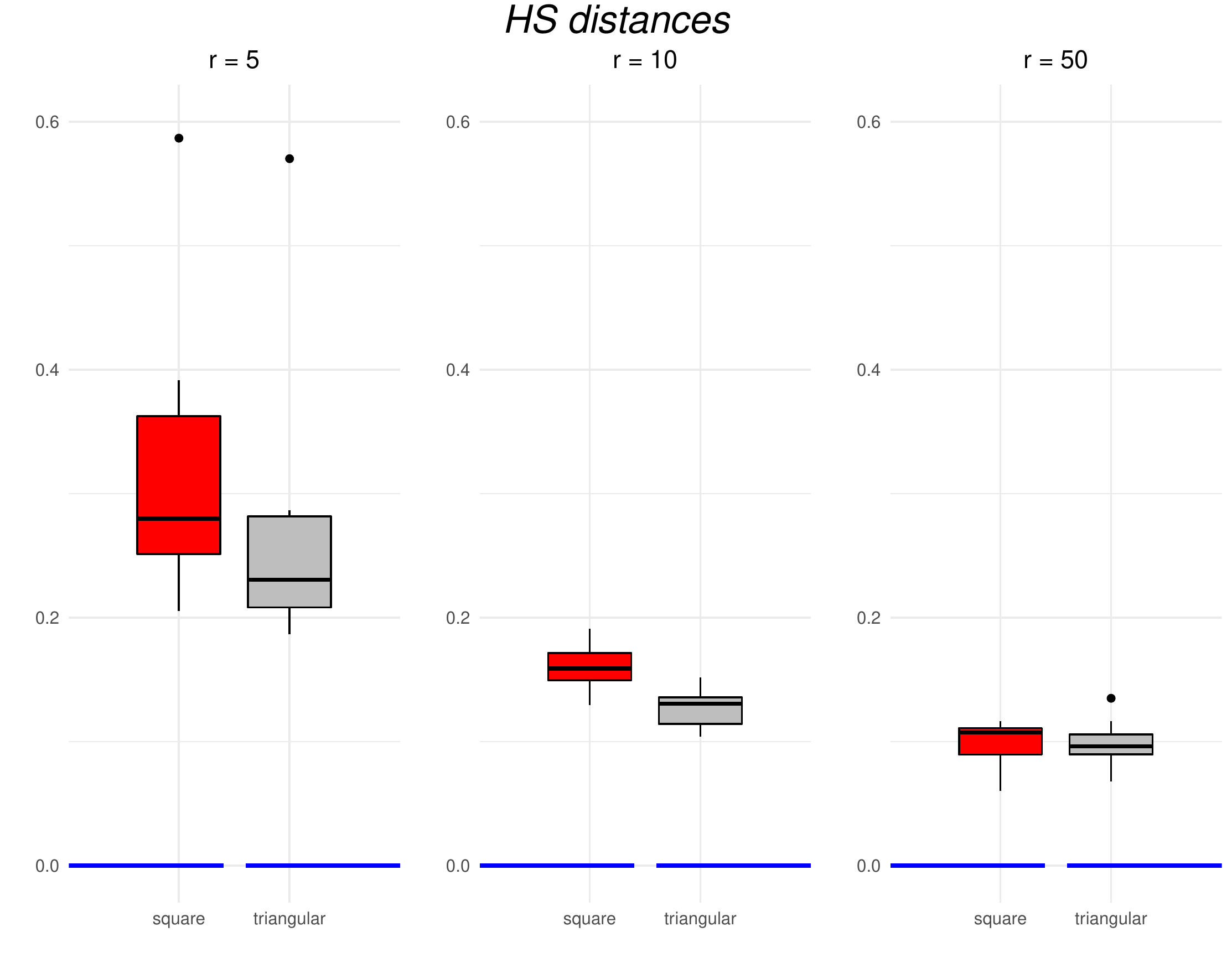}}\label{sim2:box:HS_dist}}
  \qquad
  \centering
    \subfloat[\centering ]{{\includegraphics[height=.11\textheight]{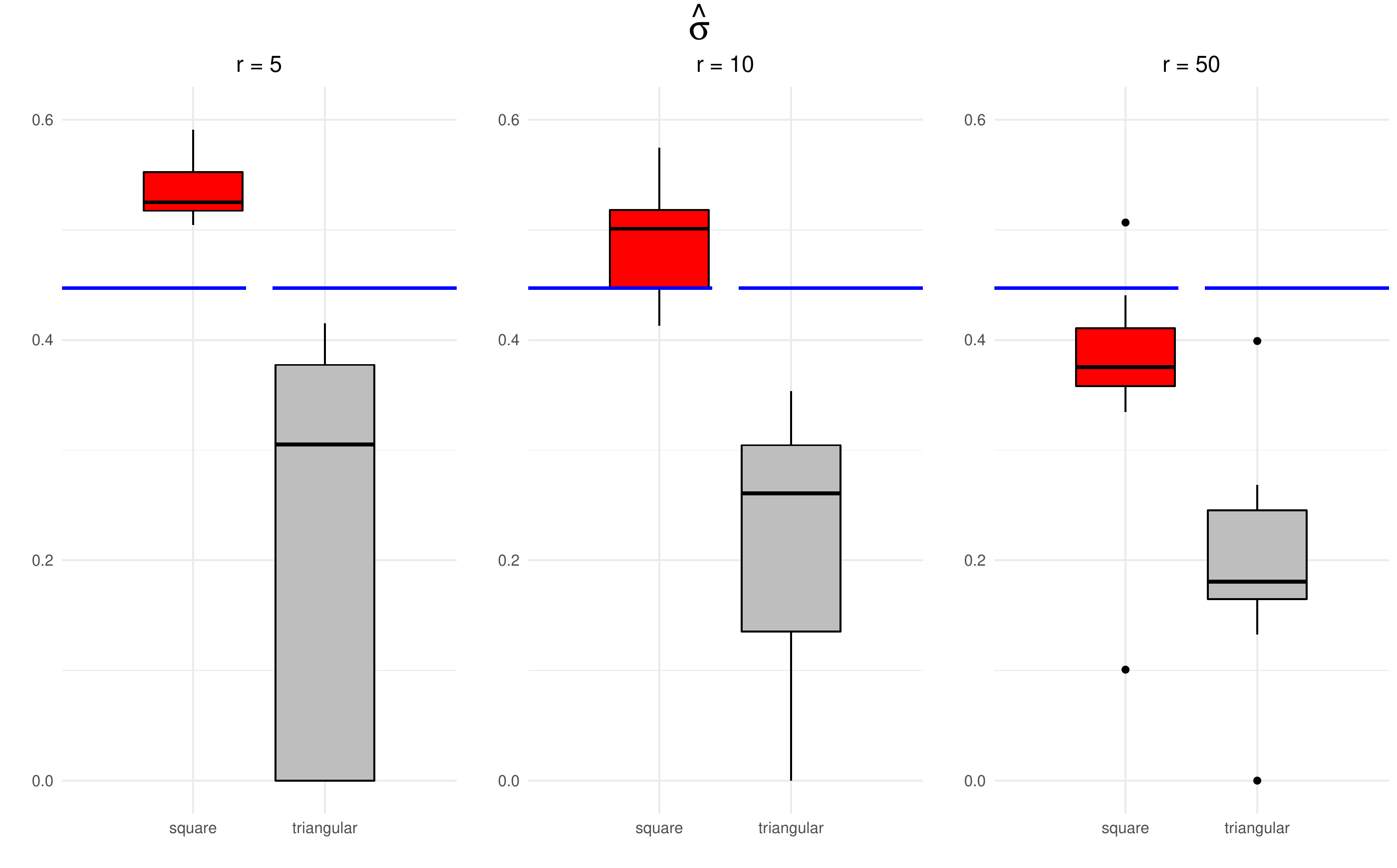}}\label{sim2:box:std_err_hat}}
    \qquad
  \centering
    \subfloat[\centering ]{{\includegraphics[height=.11\textheight]{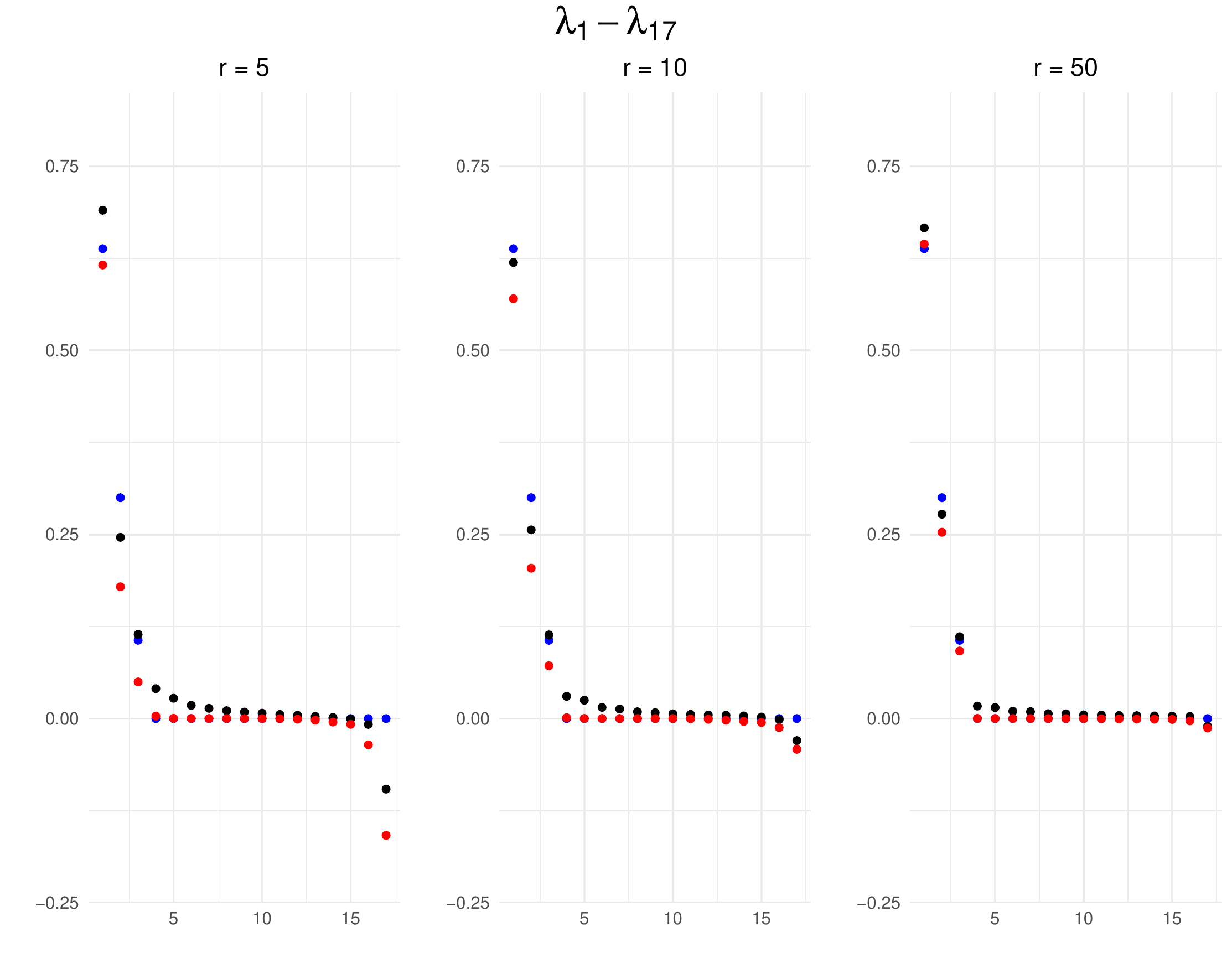}}\label{sim2:eigvals}}
  \caption{Example 2: (a) Box plots for $\Vert C - \widehat{C} \Vert_{HS}$. (b)  Box plots for $\widehat{\sigma}$. (c)   The first 17 eigenvalues, $\lambda_1$ to $ \lambda_{17}$,  and the average of their estimates over the simulation runs. }
   \end{figure}
\begin{figure}[t!]%
    \centering
    \subfloat[\centering ]{{\includegraphics[height=.11\textheight]{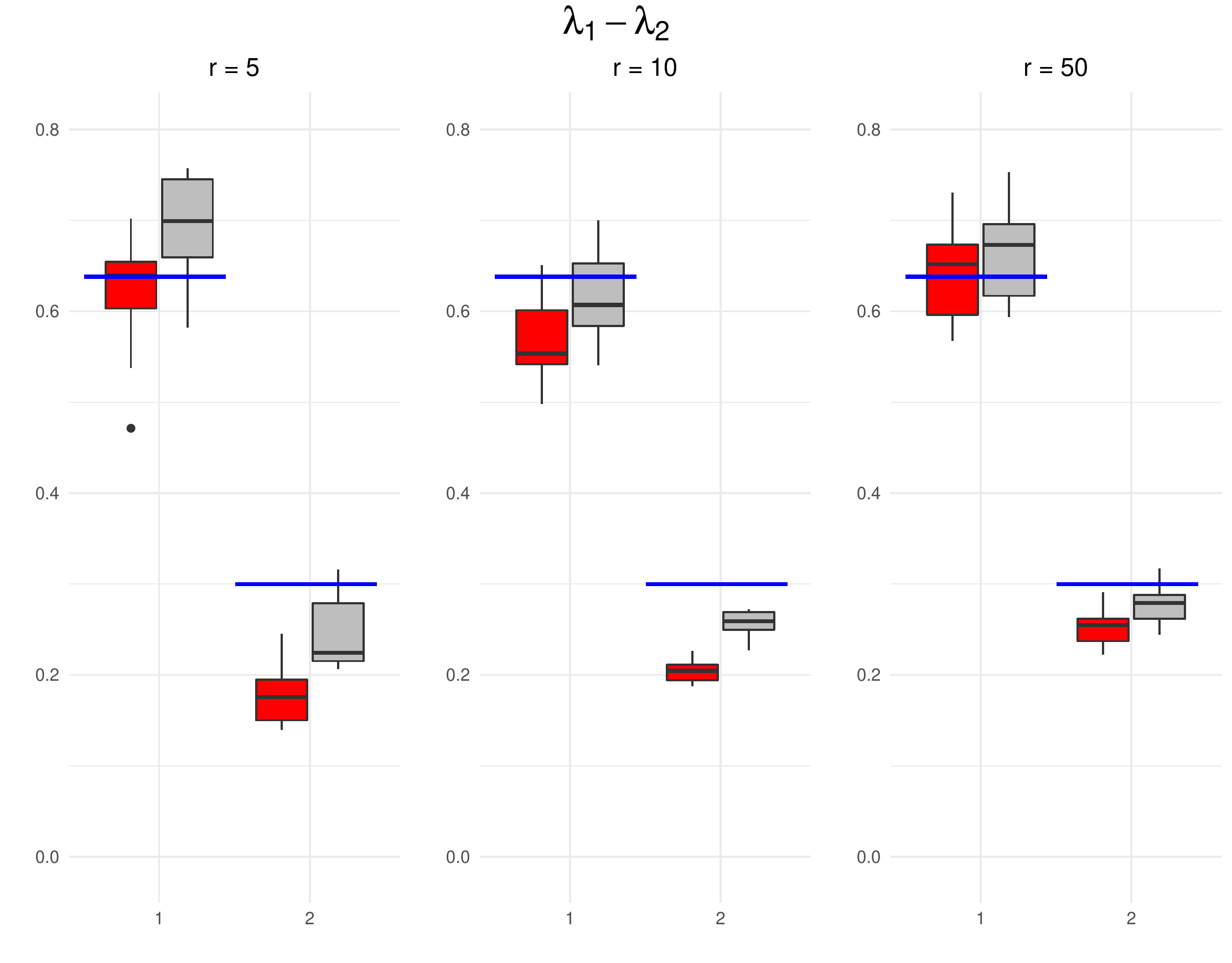} }
   \label{sim2:box:eigvals1-2}}%
    \qquad
    \subfloat[\centering ]{{\includegraphics[height=.11\textheight]{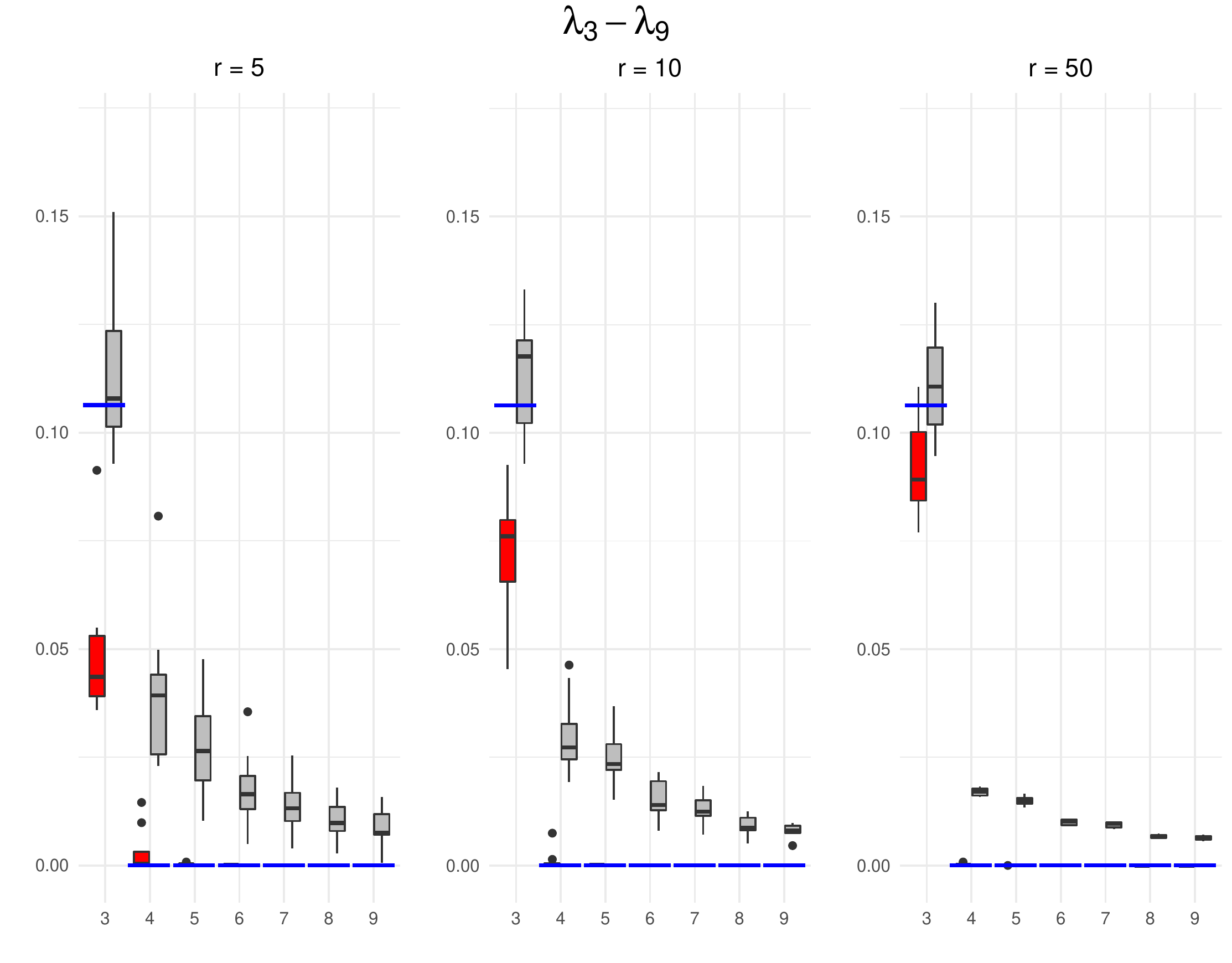} } \label{sim2:box:eigvals3-9}}%
    \quad
    \subfloat[\centering ]{{\includegraphics[height=.11\textheight]{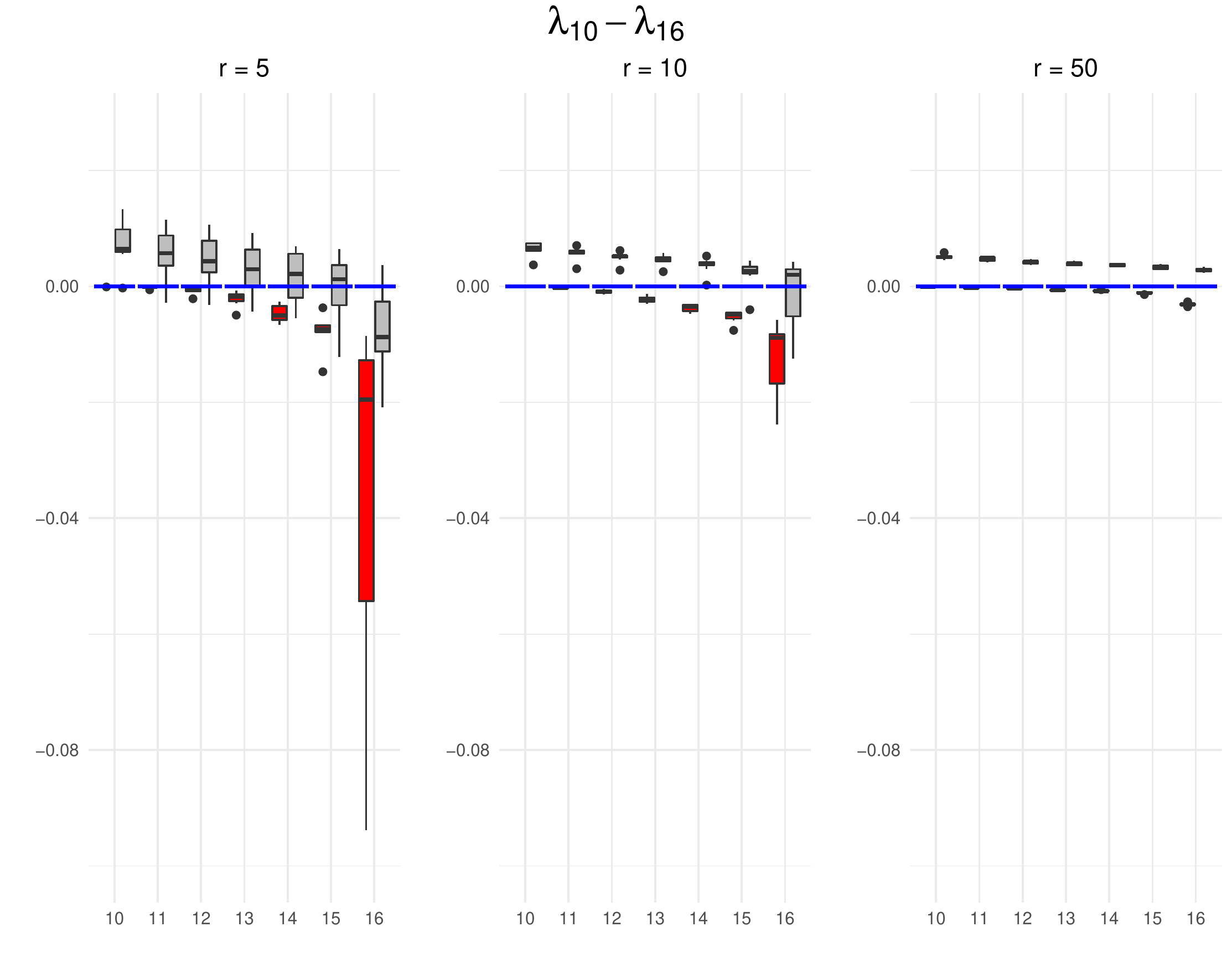} } \label{sim2:box:eigvals10-20}}%
    \caption{Example 2: (a)  Box plots for $\widehat{\lambda}_1$ and $\widehat{\lambda}_2$. (b)  Box plots for $\widehat{\lambda}_3$ to $\widehat{\lambda}_9$.(c)  Box plots for $\widehat{\lambda}_{10}$ to $\widehat{\lambda}_{16}$.}%
    \label{sim2:surfaces}%
\end{figure}

\begin{figure}[t!]%
    \centering
    \subfloat[\centering ]{{\includegraphics[height=.13\textheight]{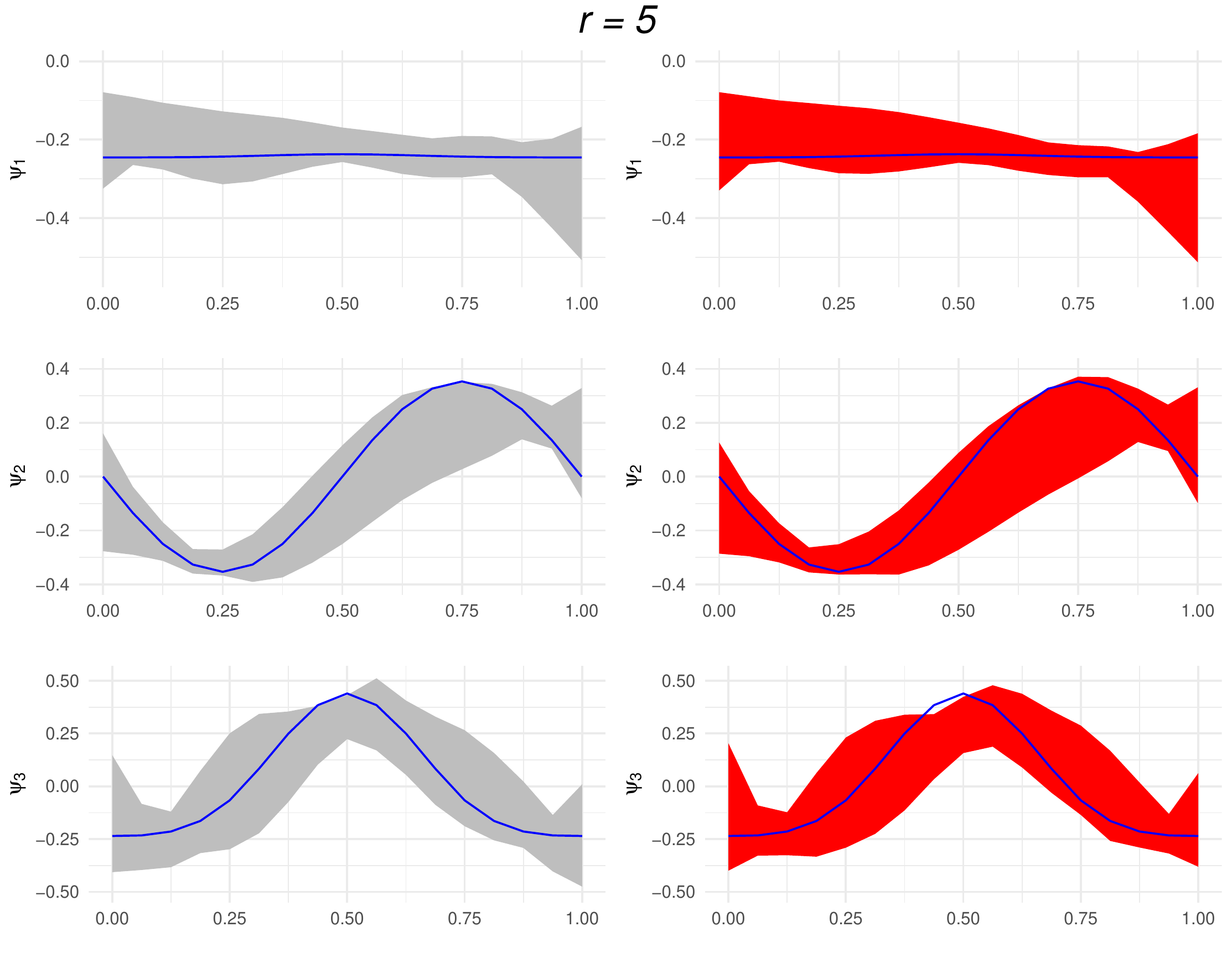} }
   }%
    \qquad
    \subfloat[\centering ]{{\includegraphics[height=.13\textheight]{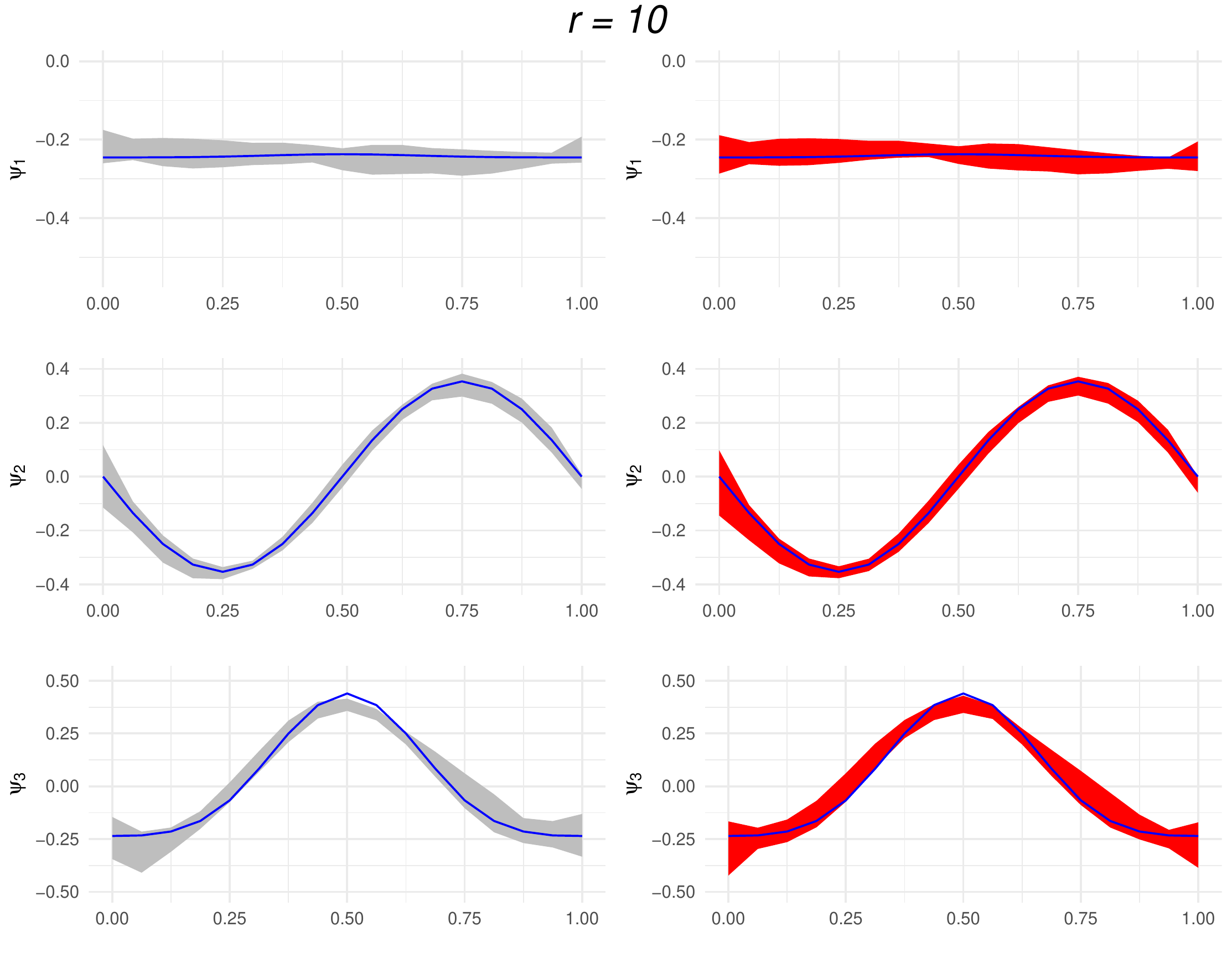} } 
    }%
    \quad
    \subfloat[\centering ]{{\includegraphics[height=.13\textheight]{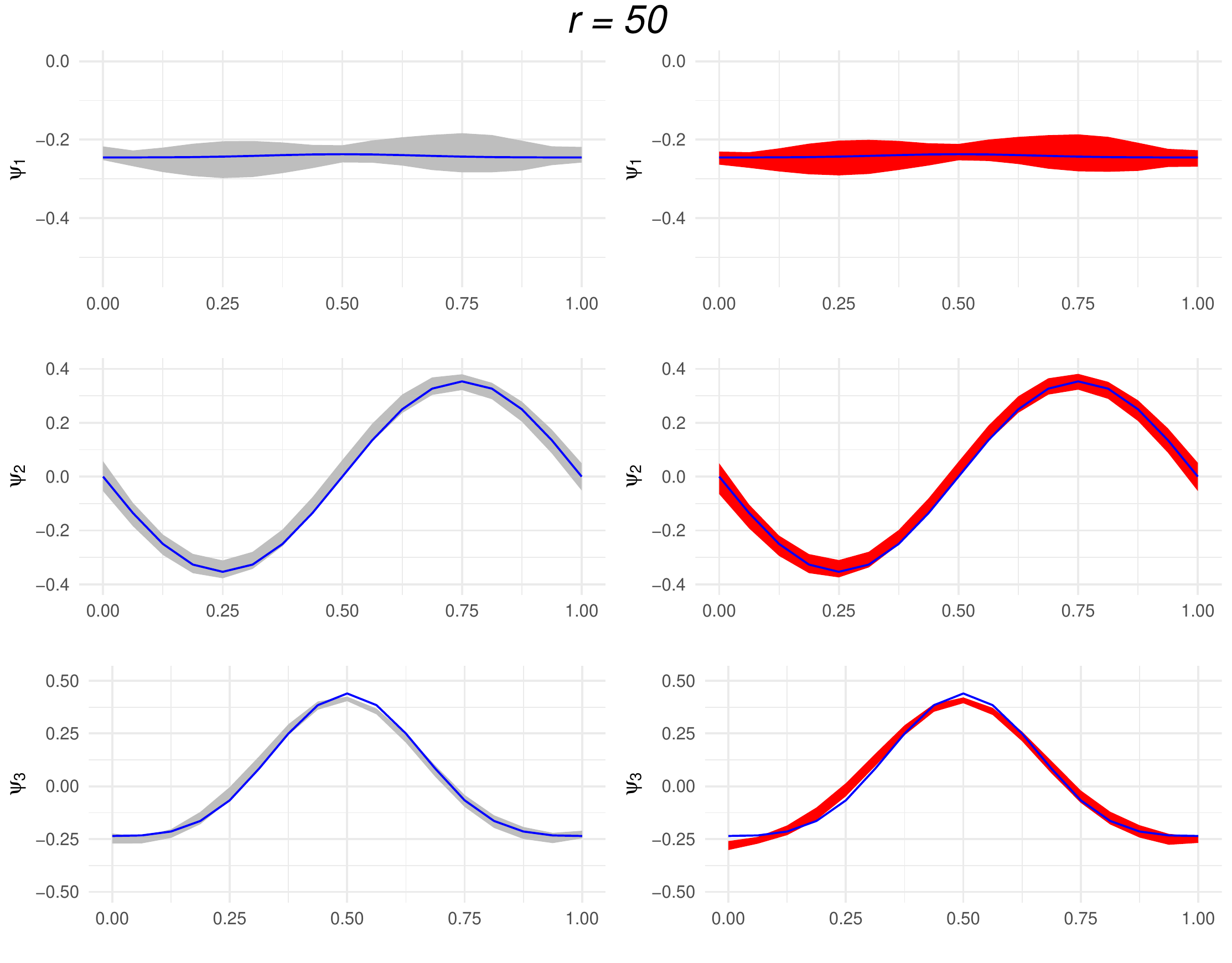}} 
    \label{sim2:eigvecs}
    }%
    \caption{Example 2, The first 4 eigenfunctions, $\Psi_1$ to $\Psi_4$, and point-wise $95\%$ confidence  envelopes based on the simulations. (a) $r=5$. (b) $r=10$. (c) $r=50$. }%
    \label{sim2:box:eigvecs}%
\end{figure}

   \begin{figure}[t!]%
    \centering
    \subfloat[\centering ]{{\includegraphics[height=.13\textheight]{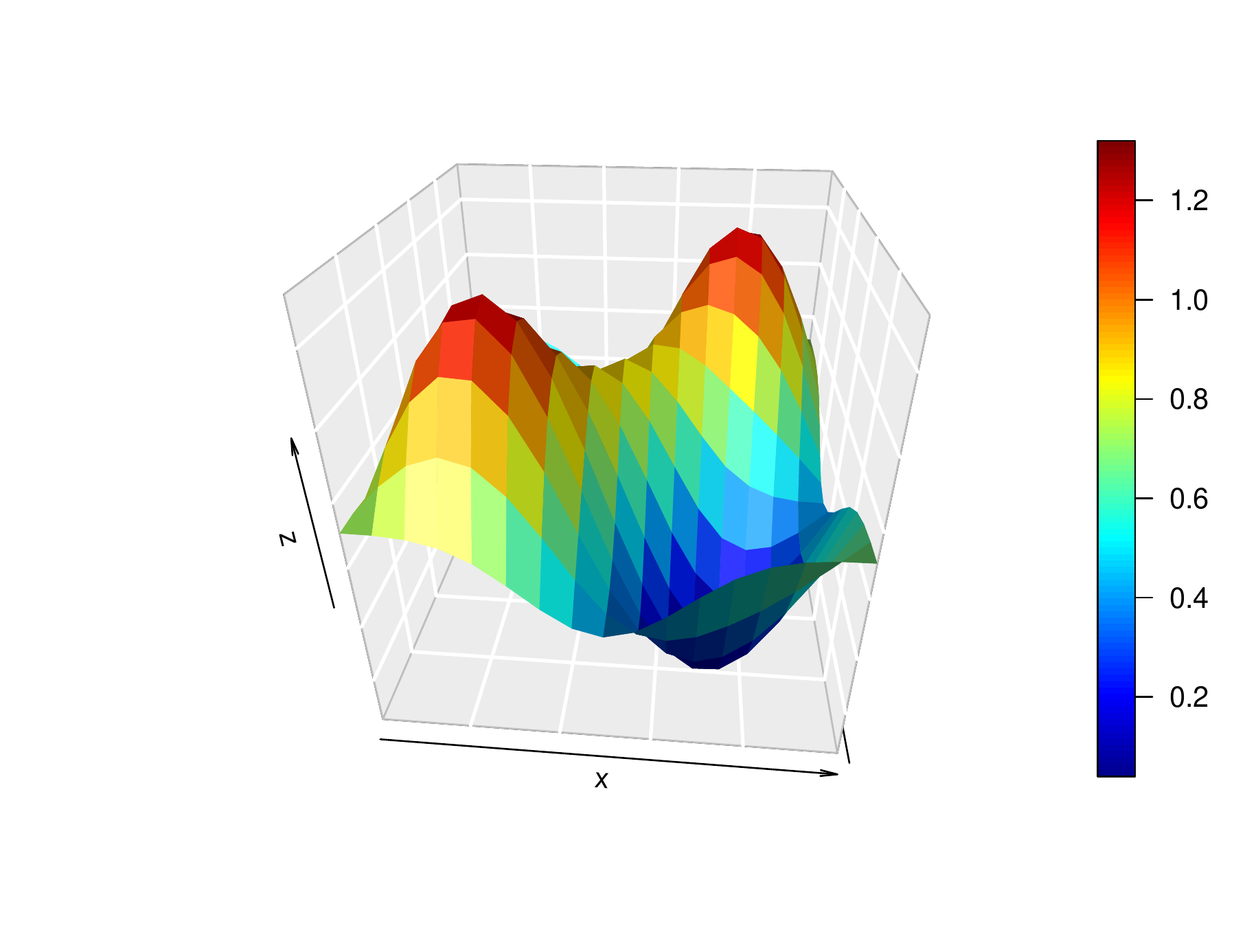} }
    \label{sim2:coc:true}}%
    \qquad
    \subfloat[\centering ]{{\includegraphics[height=.13\textheight]{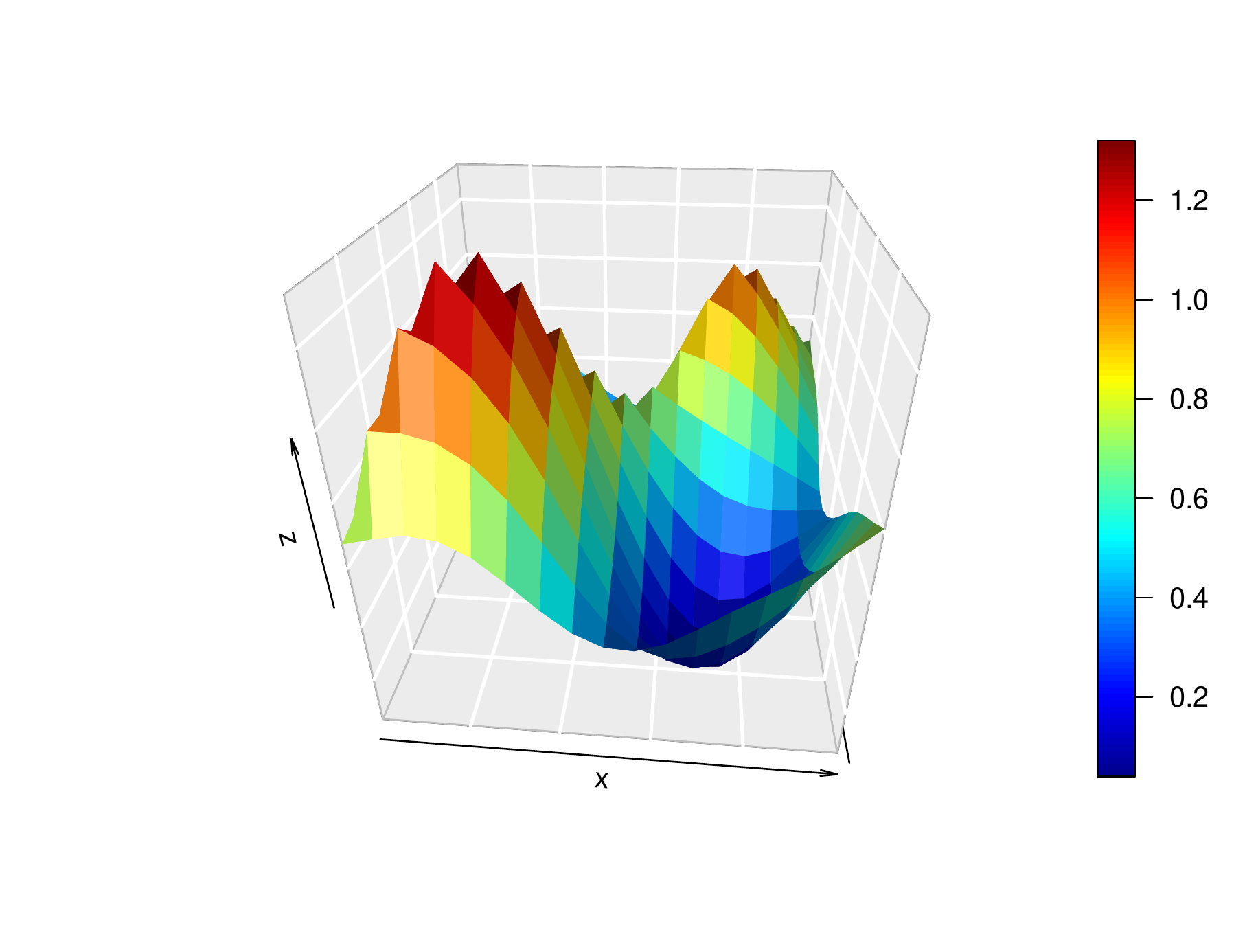} } \label{sim2:coc:triangle}}%
    \quad
    \subfloat[\centering ]{{\includegraphics[height=.13\textheight]{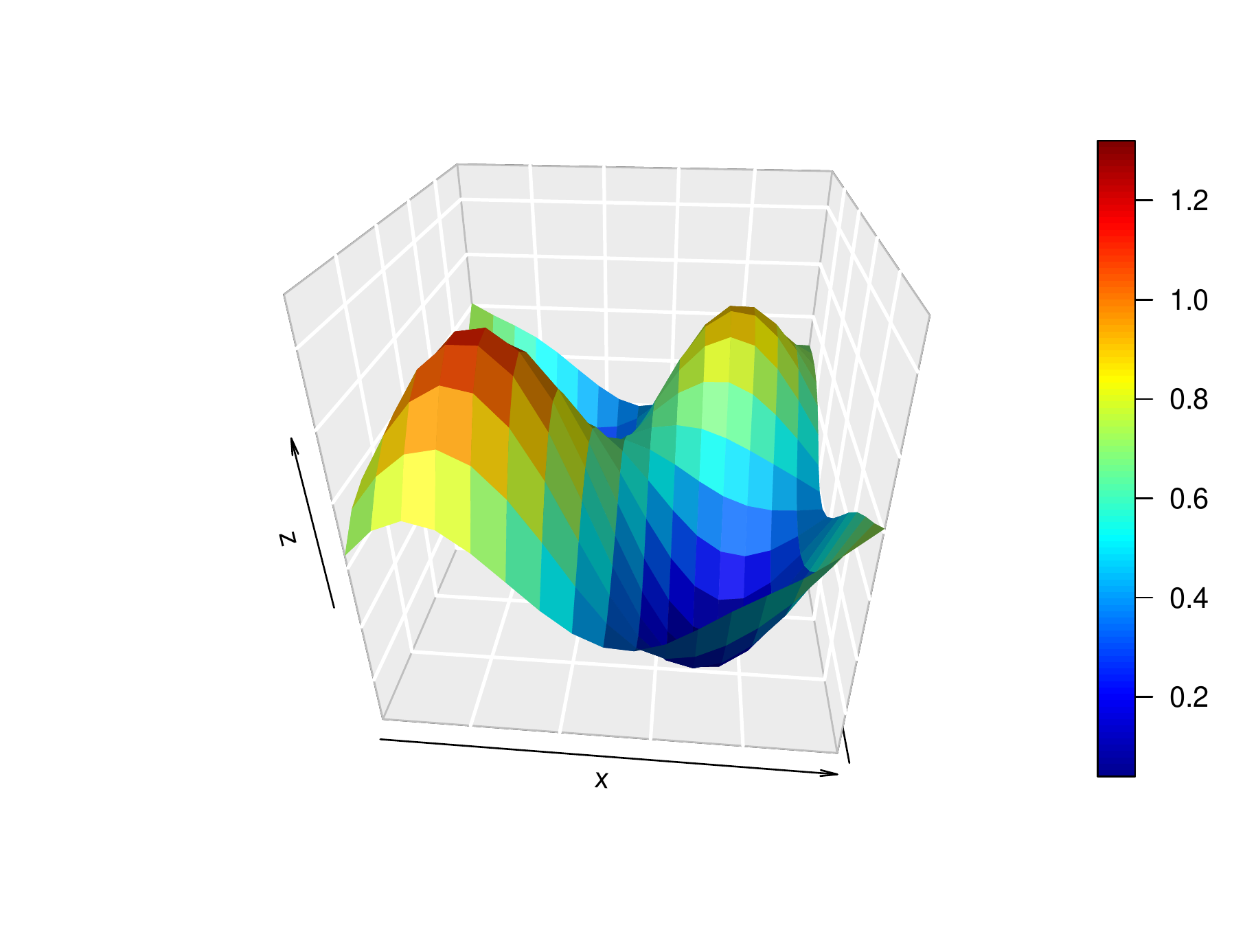} } \label{sim2:coc:square}}%
    \caption{(a): Covariance function of Example 2. (b): $\widehat{C}$, based on one of the simulation runs, applying the proposed  {reflected triangle}  approach. (c): $\widehat{C}$, based on  the same simulation run as part (b), applying the squared domain approach.}%
    \label{sim2:surfaces}%
\end{figure}
\section{Appendix} \label{Appendix}

\subsection{Proofs.}

Simple calculations shows  that  $\widehat{G}(\cdot,\cdot)$ admits the following representation 
\begin{eqnarray}\label{Ghat=A*S}
\nonumber
\widehat{G}(s,t) &=&  e_{1}^{T}
\left[
\begin{array}{ccc}
\mathcal{A}_{0,0}  & \mathcal{A}_{1,0}  &  \mathcal{A}_{0,1}\\
 \mathcal{A}_{1,0} & \mathcal{A}_{2,0}  & \mathcal{A}_{1,1} \\
 \mathcal{A}_{0,1} & \mathcal{A}_{1,1}  & \mathcal{A}_{0,2}
\end{array} 
\right]^{-1}
\left[
\begin{array}{c}
    S_{0,0}   \\
    S_{1,0}  \\
     S_{0,1} 
\end{array}
\right] 
\end{eqnarray}
which entails
\begin{align}
\widehat{G}(s,t) &=& \frac{1}{\mathcal{D}(s,t)}
\left[
\begin{array}{ccc}
 \mathcal{A}_{2,0}\mathcal{A}_{0,2}-\mathcal{A}_{1,1}^2    & - \mathcal{A}_{1,0}\mathcal{A}_{0,2}+\mathcal{A}_{0,1}\mathcal{A}_{1,1} &  \mathcal{A}_{1,0}\mathcal{A}_{1,1}-\mathcal{A}_{0,1}\mathcal{A}_{2,0}
\end{array}
\right] \left[
\begin{array}{c}
    S_{0,0}   \\
    S_{1,0}  \\
     S_{0,1} 
\end{array}
\right]
,
\end{align}
where 
\begin{eqnarray*}
 \mathcal{A}_{p,q} =\frac{1}{nh^{2}_G} \sum_{i=1}^{n} \frac{2}{r(r-1)}\sum_{1\leq k < j \leq r} W\left( \frac{T_{ij}-s}{h_G}\right) W\left( \frac{T_{ik}-t}{h_G}\right) \left( T_{ij}-s \right)^{p} \left( T_{ik}-t \right)^{q} ,\;\;\; p,q = 0,1,2,
\end{eqnarray*}
 \begin{eqnarray}
\nonumber
 S_{0,0} &=& \frac{1}{nh^2_G}\sum_{i=1}^{n}\frac{2}{r(r-1)}\sum_{1\leq k<j \leq r} \left[ W\left( \frac{T_{ij}-s}{h_G}\right) W\left( \frac{T_{ik}-t}{h_G}\right)Y_{ij}Y_{ik}\right],\\ \nonumber
 S_{0,1} &=&  \frac{1}{nh^2_G}\sum_{i=1}^{n}\frac{2}{r(r-1)}\sum_{1\leq k<j \leq r} \left[ W\left( \frac{T_{ij}-s}{h_G}\right) W\left( \frac{T_{ik}-t}{h_G}\right)
  \left(T_{ik}-t\right)Y_{ij}Y_{ik}\right],
 \\ \nonumber
 S_{1,0} &=& \frac{1}{nh^2_G}\sum_{i=1}^{n}\frac{2}{r(r-1)}\sum_{1\leq k<j \leq r} \left[ W\left( \frac{T_{ij}-s}{h_G}\right) W\left( \frac{T_{ij}-s}{h_G}\right)
  \left(T_{ij}-s\right)Y_{ij}Y_{ik}\right],
\end{eqnarray}
and
\begin{eqnarray*}
\mathcal{D}(s,t) = \det \left( \left[
\begin{array}{ccc}
\mathcal{A}_{0,0}  & \mathcal{A}_{1,0}  &  \mathcal{A}_{0,1}\\
 \mathcal{A}_{1,0} & \mathcal{A}_{2,0}  & \mathcal{A}_{1,1} \\
 \mathcal{A}_{0,1} & \mathcal{A}_{1,1}  & \mathcal{A}_{0,2}
\end{array}
\right]\right).
\end{eqnarray*}
which yields 
\begin{align}\label{G hat decomposition}
\widehat{G}(s, t) 
 &=&   \frac{\mathcal{M}_1(s,t) S_{0,0}(s,t)-\mathcal{M}_2(s,t)S_{1,0}(s,t)-\mathcal{M}_3(s,t) S_{0,1}(s,t)}{\mathcal{D}(s,t)},\;\;\;\;0<t<s<1,
\end{align}
with $\mathcal{M}_p$, $p=1,2,3$, corresponding to the elements of row matrix appearing in \eqref{Ghat=A*S}.
\\
Likewise, for  $\widehat{\mu}(\cdot)$ we have
\begin{eqnarray}\label{mu hat decomposition}
 \widehat{\mu}(s) = \frac{S_0(s)\mathcal{A}_2(s)-S_1(s)\mathcal{A}_1(s)}{\mathcal{A}_0(s)\mathcal{A}_2(s)-\mathcal{A}_1^2(s)},\;\;\; 0 < s < 1,
\end{eqnarray}
where
\begin{eqnarray}
\nonumber
 S_0 &=& \frac{1}{nh_{\mu}}\sum_{i=1}^{n}\frac{1}{r} \sum_{j=1}^{r}W\left( \frac{T_{ij}-s}{h_{\mu}}\right)Y_{ij},
 \\ \nonumber
  S_1 &=& \frac{1}{nh_{\mu}^2}\sum_{i=1}^{n}\frac{1}{r}\sum_{j=1}^{r}W\left( \frac{T_{ij}-s}{h_{\mu}}\right)\left(T_{ij}-s\right)Y_{ij},
\end{eqnarray}
and 
\begin{eqnarray*}
\mathcal{A}_p &=&   \frac{1}{nh_{\mu}^{1+p}}\sum_{i=1}^{n}\frac{1}{r} \sum_{j=1}^{r}W\left( \frac{T_{ij}-s}{h_{\mu}}\right)\left( T_{ij}-s \right)^p, \;\;\; p=0,1,2.
\end{eqnarray*}
We will make use of these expressions in our proofs.
\begin{proof}[Proof of Theorem \ref{Thm:Ghat-G}] We can drop the index $n$ for sake of simplicity. 
According to \eqref{a_0} and observing that
\begin{eqnarray*}
 G(s,t) =  e_{1}^{T}\left(\mathbb{T}_{(s,t)}^{T}\mathbf{W}_{(s,t)}\mathbb{T}_{(s,t)}\right)^{-1}\left(\mathbb{T}_{(s,t)}^{T}\mathbf{W}_{(s,t)}\mathbb{T}_{(s,t)}\right)\begin{bmatrix}
         G(s,t)\\
         G^{(1,0)}(s,t)\\
         G^{(0,1)}(s,t)
         \end{bmatrix},
\end{eqnarray*}
we obtain
\begin{align}\label{Ghat - G : decomposition}
 \underset{0 \leq t \leq s \leq 1}{\mathrm{sup}}\left|\widehat{G}(s, t) - G(s, t)\right|
 &=&  \underset{0 \leq t \leq s \leq 1}{\mathrm{sup}}\left| \frac{\mathcal{M}_1(s,t)\widetilde{S}_{0,0}(s,t)-\mathcal{M}_2(s,t)\widetilde{S}_{1,0}(s,t)-\mathcal{M}_3(s,t)\widetilde{S}_{0,1}(s,t)}{\mathcal{D}(s,t)}\right|,
 \end{align}
 where
\begin{eqnarray}
\nonumber
 \widetilde{S}_{0,0}(s,t) &=&S_{0,0}(s,t) - G(s,t) \mathcal{A}_{0,0}-h_G G^{(1,0)}(s,t) \mathcal{A}_{1,0} -h_G G^{(0,1)}(s,t) \mathcal{A}_{0,1}\\ \nonumber
 &=:& S_{0,0}(s,t) - \Lambda_{0,0}(s,t),
 \\ \nonumber
 \widetilde{S}_{1,0}(s,t) &=&S_{1,0}(s,t) - G(s,t) \mathcal{A}_{1,0}-h_G G^{(1,0)}(s,t) \mathcal{A}_{2,0} -h_G G^{(0,1)}(s,t) \mathcal{A}_{1,1}  \\ \nonumber
 &=:& S_{1,0}(s,t) - \Lambda_{1,0}(s,t), \\ \nonumber
  \widetilde{S}_{0,1}(s,t) &=&S_{0,1}(s,t) - G(s,t) \mathcal{A}_{0,1}-h_G G^{(1,0)}(s,t) \mathcal{A}_{1,1} -h_G G^{(0,1)}(s,t) \mathcal{A}_{0,2}\\ \nonumber
  &=:& S_{0,1}(s,t) - \Lambda_{0,1}(s,t).
\end{eqnarray}
We now investigate the different terms appearing in \eqref{Ghat - G : decomposition} separately. First, consider $\widetilde{S}_{0,0}$ and observe that 
\begin{eqnarray*}
  \widetilde{S}_{0,0}(s,t) &=&\frac{1}{nh^2_G}\sum_{i=1}^{n}\frac{2}{r(r-1)}\sum_{1\leq k<j \leq r} \left[ W\left( \frac{T_{ij}-s}{h_G}\right) W\left( \frac{T_{ik}-t}{h_G}\right)Y_{ij}Y_{ik}\right] - \Lambda_{0,0}(s,t)\\ 
  &=& \frac{1}{nh^2_G}\sum_{i=1}^{n}\frac{2}{r(r-1)}\sum_{1\leq k<j \leq r} \left[ W\left( \frac{T_{ij}-s}{h_G}\right) W\left( \frac{T_{ik}-t}{h_G}\right)\left(X_{ij} + U_{ij}\right)\left(X_{ik} + U_{ik}\right)\right] \\ 
  &&- \Lambda_{0,0}(s,t)\\
  &=&\frac{1}{nh^2_G}\sum_{i=1}^{n}\frac{2}{r(r-1)}\sum_{1\leq k<j \leq r} \left[ W\left( \frac{T_{ij}-s}{h_G}\right) W\left( \frac{T_{ik}-t}{h_G}\right) U_{ij}U_{ik}\right]\\
  &&+\frac{1}{nh^2_G}\sum_{i=1}^{n}\frac{2}{r(r-1)}\sum_{1\leq k<j \leq r} \left[ W\left( \frac{T_{ij}-s}{h_G}\right) W\left( \frac{T_{ik}-t}{h_G}\right) X_{ij}U_{ik}\right]\\
  &&+\frac{1}{nh^2_G}\sum_{i=1}^{n}\frac{2}{r(r-1)}\sum_{1\leq k<j \leq r} \left[ W\left( \frac{T_{ij}-s}{h_G}\right) W\left( \frac{T_{ik}-t}{h_G}\right) U_{ij}X_{ik}\right]\\
  &&+\frac{1}{nh^2_G}\sum_{i=1}^{n}\frac{2}{r(r-1)}\sum_{1\leq k<j \leq r} \left[ W\left( \frac{T_{ij}-s}{h_G}\right) W\left( \frac{T_{ik}-t}{h_G}\right) \left(X_{ij}X_{ik}-G\left(T_{ij},T_{ik} \right)\right)\right]\\
  &&+\frac{1}{nh^2_G}\sum_{i=1}^{n}\frac{2}{r(r-1)}\sum_{1\leq k<j \leq r} \left[ W\left( \frac{T_{ij}-s}{h_G}\right) W\left( \frac{T_{ik}-t}{h_G}\right) G\left(T_{ij},T_{ik} \right)\right]- \Lambda_{0,0}(s,t)\\
  &=:& A_1+A_2+A_3+A_4+A_5.
\end{eqnarray*}
The expressions $A_1-A_4$, representing the variance term, can be written in the general form
\begin{align}
\nonumber 
 &\frac{1}{nh^2_G}\sum_{i=1}^{n}\frac{2}{r(r-1)}\sum_{1\leq k<j \leq r} \left[ W\left( \frac{T_{ij}-s}{h_G}\right) W\left( \frac{T_{ik}-t}{h_G}\right) Z_{ijk}\right] \\  \nonumber 
&=  \frac{1}{nh^2_G}\sum_{i=1}^{n}\frac{2}{r(r-1)}\sum_{1\leq k<j \leq r} \left[ W\left( \frac{T_{ij}-s}{h_G}\right) W\left( \frac{T_{ik}-t}{h_G}\right) Z_{ijk}\right] \\ 
\label{sum:Z_ijk:c:c} 
& \hspace{4.1cm}\times \left[\mathbb{I}\left(  \left(T_{ij}, T_{ik}\right) \in [s-h_G , s+h_G]^{c} \times [t-h_G , t+h_G]^{c} \right)\right.\\  
\label{sum:Z_ijk:c:} 
& \hspace{4.5cm}+ \left.\mathbb{I}\left(  \left(T_{ij}, T_{ik}\right) \in [s-h_G , s+h_G]^{c} \times [t-h_G , t+h_G] \right)\right.\\  \label{sum:Z_ijk::c} 
& \hspace{4.5cm}+ \left.\mathbb{I}\left(  \left(T_{ij}, T_{ik}\right) \in [s-h_G , s+h_G] \times [t-h_G , t+h_G]^{c} \right)\right.\\ 
 \label{sum:Z_ijk} 
& \hspace{4.5cm}+ \left.\mathbb{I}\left(  \left(T_{ij}, T_{ik}\right) \in [s-h_G , s+h_G] \times [t-h_G , t+h_G] \right)\right]\\
\nonumber
&=: \bar{Z}_{1,1} (s,t) +\bar{Z}_{1,0} (s,t)+\bar{Z}_{0,1} (s,t)+\bar{Z}_{0,0} (s,t),
\end{align}
where each $Z_{ijk}$ has mean zero. The term $A_5$ corresponds to the bias term for which we use a Taylor series expansion to obtain the desired almost sure uniform bound. For \eqref{sum:Z_ijk:c:c}, observe that
\begin{align}
\label{Z_11}
\bar{Z}_{1,1} (s,t) &= \frac{1}{nh^2_G}\sum_{i=1}^{n}\frac{2}{r(r-1)}\sum_{1\leq k<j \leq r} \left[ W\left( \frac{T_{ij}-s}{h_G}\right) W\left( \frac{T_{ik}-t}{h_G}\right) Z_{ijk}\right] \\  \nonumber
 &\hspace{4.1cm}\times \mathbb{I}\left(  \left(T_{ij}, T_{ik}\right) \in [s-h_G , s+h_G]^{c} \times [t-h_G , t+h_G]^{c} \right)\\ \nonumber
 &\leq  W^{2}\left( 1^{+} \right)\frac{1}{nh^2_G}\sum_{i=1}^{n}\frac{2}{r(r-1)}\sum_{1\leq k<j \leq r} \vert Z_{ijk} \vert \\ \nonumber
&= O\left( h_G^{-2} \right) W^{2}\left( 1^{+} \right),\;\;\;\; \mathrm{\;a.s\;uniformly\; on\;\;}  0 \leq t \leq s \leq 1 \\
&=  O\left( h_G^{2} \right) ,\;\;\;\; \mathrm{\;a.s\;uniformly\; on\;\;}  0 \leq t \leq s  \leq 1.
\end{align}
For \eqref{sum:Z_ijk:c:} (similarly \eqref{sum:Z_ijk::c}) we have
\begin{align}
\label{Z_10}
\bar{Z}_{1,0} (s,t) &= \frac{1}{nh^2_G}\sum_{i=1}^{n}\frac{2}{r(r-1)}\sum_{1\leq k<j \leq r} \left[ W\left( \frac{T_{ij}-s}{h_G}\right) W\left( \frac{T_{ik}-t}{h_G}\right) Z_{ijk}\right] \\  \nonumber
&\hspace{4.1cm} \times \mathbb{I}\left(  \left(T_{ij}, T_{ik}\right) \in [s-h_G , s+h_G]^{c} \times [t-h_G , t+h_G] \right)\\ \nonumber
&\leq  \left(\int \left\vert W^{\dagger}\left( u\right)\right\vert du \right)  W\left( 1^{+} \right)   \frac{1}{nh^2_G}\sum_{i=1}^{n}\frac{2}{r(r-1)}\sum_{1\leq k<j \leq r} \vert Z_{ijk} \vert\\ \nonumber
&= O\left( h_G^{-2} \right) \left(\int \left\vert W^{\dagger}\left( u\right)\right\vert du \right)  W\left( 1^{+} \right),\;\;\;\; \mathrm{\;a.s\;uniformly\; on\;\;}  0 \leq t \leq s  \leq 1 \\
&=  O\left( h_G^{2} \right) ,\;\;\;\; \mathrm{\;a.s\;uniformly\; on\;\;}  0 \leq t \leq s  \leq 1.
\end{align}
For \eqref{sum:Z_ijk} we have 
\begin{align}
\nonumber
\bar{Z}_{0,0} (s,t)&=&
 \\ \label{int:int:W:W}
 &&  \frac{1}{n}\sum_{i=1}^{n}\frac{2}{r(r-1)}\sum_{1\leq k<j \leq r} \left[ Z_{ijk} \iint e^{-\mymathbb{i}us-\mymathbb{i}vt+\mymathbb{i}uT_{ij}+\mymathbb{i}vT_{ik}}W^{ \dagger} (h_Gu)W^{ \dagger} (h_Gv)dudv\right]  \\  \nonumber
&& \hspace{4cm}\times \mathbb{I}\left(  \left(T_{ij}, T_{ik}\right) \in [s-h_G , s+h_G] \times [t-h_G , t+h_G] \right).
\end{align}
Regarding $\mathbb{E}\bar{Z}_{0,0}(s,t) = 0 $, for all $0 \leq t \leq s \leq 1$, we have
 \begin{align}
  \nonumber
 &\vert \bar{Z}_{0,0} (s,t)- \mathbb{E}\bar{Z}_{0,0}(s,t)\vert \\ \label{int:W}
 &\leq \left(\int \left\vert W^{\dagger}\left( h_G u\right)\right\vert du \right)^{2}
 \\ \nonumber & \times
 \frac{1}{n}\sum_{i=1}^{n}\frac{2}{r(r-1)}\sum_{1\leq k<j \leq r}\left\vert Z_{ijk} \right\vert \mathbb{I}\left(  \left(T_{ij}, T_{ik}\right) \in [s-h_G , s+h_G] \times [t-h_G , t+h_G] \right)
 \\ \label{rate*h^4}
 & =  O\left( h^{-2}_{G}\right)
  \times
 \frac{1}{n}\sum_{i=1}^{n}\frac{2}{r(r-1)}\sum_{1\leq k<j \leq r}\left\vert Z_{ijk} \right\vert \mathbb{I}\left(  \left(T_{ij}, T_{ik}\right) \in [s-h_G , s+h_G] \times [t-h_G , t+h_G] \right).
 \end{align}
The remaining argument for \eqref{sum:Z_ijk} is similar to the proof of Lemma (8.2.5) of \citet{hsing_theoretical_2015}. In more detail, for the summation term appearing in \eqref{rate*h^4} we have
 \begin{eqnarray*}
  &&\frac{1}{n}\sum_{i=1}^{n}\frac{2}{r(r-1)}\sum_{1\leq k<j \leq r}\left\vert Z_{ijk} \right\vert \mathbb{I}\left(  \left(T_{ij}, T_{ik}\right) \in [s-h_G , s+h_G] \times [t-h_G , t+h_G] \right)\\
  &=& \frac{1}{n}\sum_{i=1}^{n}\frac{2}{r(r-1)}\sum_{1\leq k<j \leq r}  
  \left\vert Z_{ijk}  \right\vert    \mathbb{I}\left(   \left\vert Z_{ijk}  \right\vert  \geq   Q_n \cup \left\vert Z_{ijk}  \right\vert  < Q_n   \right)\\
  && \hspace{4cm} \times \mathbb{I}\left(  \left(T_{ij}, T_{ik}\right) \in [s-h_G , s+h_G] \times [t-h_G , t+h_G] \right)\\
  &=:& B_1 + B_2.
  \end{eqnarray*}
According to condition \ref{sup|.|^a} and choosing $Q_n$ in such a way that $\left[ \frac{\mathrm{log}n}{n}\left( h_G^4 + \frac{h_G^3}{r(n)}+ \frac{h_G^2}{r^2(n)} \right)\right]^{-1/2} Q_n^{1-\alpha} = O(1)$, we conclude $B_1 = O \left( \left[ \frac{\mathrm{log}n}{n}\left( h_G^4 + \frac{h_G^3}{r(n)}+ \frac{h_G^2}{r^2(n)} \right)\right]^{1/2} \right)$, almost surely uniformly. In more detail
\begin{eqnarray}
\nonumber
 B_1  & = & \frac{1}{n}\sum_{i=1}^{n}\frac{2}{r(r-1)}\sum_{1\leq k<j \leq r}
  \left\vert Z_{ijk}  \right\vert ^{1-\alpha+\alpha}   \mathbb{I}\left(   \left\vert Z_{ijk}  \right\vert  \geq   Q_n  \right)  \mathbb{I}\left(  \left(T_{ij}, T_{ik}\right) \in [s-h_G , s+h_G] \times [t-h_G , t+h_G] \right)\\ \nonumber
  & \leq & \frac{1}{n}\sum_{i=1}^{n}\frac{2}{r(r-1)}\sum_{1\leq k<j \leq r}
  \left\vert Z_{ijk}  \right\vert ^{\alpha}  Q_n^{1-\alpha}   \\ \nonumber
  & = & O(1) Q_n^{1-\alpha} \\
  &=&   \nonumber
  O\left( \left[ \frac{\mathrm{log}n}{n}\left( h_G^4 + \frac{h_G^3}{r(n)}+ \frac{h_G^2}{r^2(n)} \right)\right]^{1/2} \right),\;\;\;\; \mathrm{\;a.s\;uniformly\; on\;}  0<t<s<1.
  \end{eqnarray}
For $B_2$, first, define
  \begin{eqnarray}
  \nonumber
  B_2 &=& \frac{1}{n}\sum_{i=1}^{n}\frac{2}{r(r-1)}\sum_{1\leq k<j \leq r}  
  \left\vert Z_{ijk}  \right\vert    \mathbb{I}\left(  \left\vert Z_{ijk}  \right\vert  \leq Q_n   \right)
 \times \mathbb{I}\left(  \left(T_{ij}, T_{ik}\right) \in [s-h_G , s+h_G] \times [t-h_G , t+h_G] \right)\\ \nonumber
 &=:& \frac{1}{n}\sum_{i=1}^{n}\frac{2}{r(r-1)}\sum_{1\leq k<j \leq r}  \mathcal{Z}_{ijk}(s,t).
  \end{eqnarray}
We then apply Bennett's concentration inequality (see \citep{boucheron_concentration_2013}) to complete the proof of this part. We  obtain a uniform upper bound for  $\mathrm{Var\left( \frac{2}{r(r-1)}\sum_{1\leq k<j \leq r}  \mathcal{Z}_{ijk}(s,t) \right) }$, $i=1,2,\ldots,n$, in the following way
\begin{eqnarray}
\nonumber
 &&\mathrm{Var\left( \frac{2}{r(r-1)}\sum_{1\leq k<j \leq r}  \mathcal{Z}_{ijk}(s,t) \right) } 
 \\ \nonumber &=& \left(\frac{2}{r(r-1)}\right)^{2}
 \sum_{1\leq k_1<j_1 \leq r}  \sum_{1\leq k_2<j_2 \leq r} \mathrm{Cov}\left( \mathcal{Z}_{ij_1k_1}(s,t) , \mathcal{Z}_{ij_2k_2}(s,t) \right)\\ \label{var:Bern}
 & \leq & c \left( h_G^4 + \frac{h_G^3}{r(n)}+ \frac{h_G^2}{r^2(n)} \right),
\end{eqnarray}
for some positive constant $c$ which  depends neither on $(s,t)$ nor on $i$. Note that inequality \eqref{var:Bern} is a direct consequence of conditions \ref{sup|.|^a}  and \ref{bound:density}. 
Finally, applying Bennett's inequality and choosing $Q_n = \left(\frac{\mathrm{log}n}{n}\right)^{-1/2}\left( h_G^4 + \frac{h_G^3}{r(n)}+ \frac{h_G^2}{r^2(n)} \right)^{1/2}$ we have, for any positive number $\eta$,
\begin{align}
\nonumber
 &\mathbb{P}\left( \frac{1}{n}\sum_{i=1}^{n}\frac{2}{r(r-1)}\sum_{1\leq k<j \leq r}  \mathcal{Z}_{ijk}(s,t) \geq \eta \left[ \frac{\mathrm{log}n}{n}\left( h_G^4 + \frac{h_G^3}{r(n)}+ \frac{h_G^2}{r^2(n)} \right) \right]^{1/2} \right) \\ \nonumber
 &\leq\mathrm{exp}\left\{ -\frac{\eta^2 n^2 \left[ \frac{\mathrm{log}n}{n}\left( h_G^4 + \frac{h_G^3}{r(n)}+ \frac{h_G^2}{r^2(n)} \right) \right]}{2nc\left( h_G^4 + \frac{h_G^3}{r(n)}+ \frac{h_G^2}{r^2(n)} \right)+\frac{2}{3}\eta n\left( h_G^4 + \frac{h_G^3}{r(n)}+ \frac{h_G^2}{r^2(n)} \right)} \right\} \\ \nonumber
 &= \mathrm{exp}\left\{-\frac{\eta^2   \mathrm{log}n }{ 2c+\frac{2}{3}\eta} \right\}\\ \label{sum:Bern}
 &= n^{-\frac{\eta^2 }{ 2c+\frac{2}{3}\eta} }\;, \;\;\;\; \forall 0 \leq t \leq s \leq 1.
\end{align}
Choosing $\eta$ large enough we conclude summability of \eqref{sum:Bern}. This result together combined with the Borel-Cantelli lemma completes the proof of this part. In other words we conclude there exists a subset $\Omega_0 \subset \Omega$ of full probability measure such that for each $\omega \in \Omega_0$ there exists $n_0 = n_0(\omega) $ with 
\begin{align}\label{rate:varianc}
 \frac{1}{n}\sum_{i=1}^{n}\frac{2}{r(r-1)}\sum_{1\leq k<j \leq r}  \mathcal{Z}_{ijk}(s,t) \leq \eta \left[ \frac{\mathrm{log}n}{n}\left( h_G^4 + \frac{h_G^3}{r(n)}+ \frac{h_G^2}{r^2(n)} \right) \right]^{1/2},\;\;\;\; n \geq n_0.
\end{align}
We now turn to the bias term $A_5$ and investigate its convergence. Observe that
\begin{align}
\nonumber
 A_5 &= \frac{1}{nh^2_G}\sum_{i=1}^{n}\frac{2}{r(r-1)}\sum_{1\leq k<j \leq r} \left[ W\left( \frac{T_{ij}-s}{h_G}\right) W\left( \frac{T_{ik}-t}{h_G}\right) G\left(T_{ij},T_{ik} \right)\right]- \Lambda_{0,0}(s,t)\\ \nonumber
 &=\frac{1}{nh^2_G}\sum_{i=1}^{n}\frac{2}{r(r-1)}\sum_{1\leq k<j \leq r} W\left( \frac{T_{ij}-s}{h_G}\right) W\left( \frac{T_{ik}-t}{h_G}\right)\left[  G\left(T_{ij},T_{ik} \right) - G(s,t) \right.\\ \nonumber
 & \hspace{5cm}
 \left.
 -\left( T_{ij}-s\right) G^{(1,0)}(s,t) - \left( T_{ik}-t\right) G^{(0,1)}(s,t)\right]\\ \label{rate:bias}
 & =  O\left( h^2_G \right),\;\;\;\; \mathrm{a.s\;uniformly \; on \;} 0 \leq t \leq s \leq 1.
\end{align}
The last expression is a consequence of the smoothness condition \ref{2-diff:C}. Relations \eqref{rate:varianc} and \eqref{rate:bias} complete the proof for  term $\widetilde{S}_{0,0}$ appearing in decomposition \eqref{Ghat - G : decomposition}. The proofs for the two other terms $\widetilde{S}_{1,0}$ and $\widetilde{S}_{0,1}$ are similar except for $W^\dagger(\cdot)W^\dagger(\cdot)$ appearing in \eqref{Z_11} and \eqref{Z_10} for each of which we may use either $W^\dagger(\cdot)\Xi_p^\dagger(\cdot)$, $p=1,2,3$,  instead. 
One may also easily modify the arguments for \eqref{int:int:W:W} and \eqref{int:W}. The result for the   terms $\mathcal{M}_1$, $\mathcal{M}_2$, $\mathcal{M}_3$, and $\mathcal{D}$ can be obtain by a slight modification of the above  argument assigning $Z_{ijk} = 1$ in relation \eqref{sum:Z_ijk}.
\end{proof}


\subsection{Choice of Kernel Function}

In the formulation of our methods and asymptotic theory, we recommended the use of two specific kernel functions, whose form evolves with $n$, namely $ \mathbf{W}_n(u)$ (equation \eqref{kernel1}) and  $\mathbb{W}_n(u)$ (equation \eqref{kernel2}). These differ from a more standard choice of a smooth compactly supported kernel function whose form is fixed with respect to $n$. The reason for this is technical: with such a choice, there is a positive probability that the estimators  {$\widehat{\mu}$ and $\widehat{G}$ }
may fail to exist on a set of positive Lebesgue measure, regardless of sample size. To see this, consider the mean estimator, and a standard design where the $\{T_{ij}\}$ are independent and uniform on $[0,1]$. In such a case, the support of the mean is sampled at the $nr_n$ independent and uniform locations, and writing $T_{(1)}=\min_{i\leq n,j\leq r_n} T_{ij}$,  we have $\mathbb{P}\{T_{(1)}>h\} = (1-h)^{nr_n}>0$, for all $n\geq 1$. On this event, the estimator is undefined on an open set near the left boundary. This is not just a boundary issue: it can also occur in the ``interior" part of domain of the mean function, i.e. it can be shown that $\mathbb{P}\Big\{\cup_{j=2}^{nr_n} \{T_{(j)} - T_{(j-1)}>2h\} \Big\}>0$, where $\{T_{(j)}\}_{1\leq j \leq nr_n}$ is the sequence of ordered statistics of $\{T_{ij}\}$.
To circumvent this problem, we choose a sequence such that $W_n(u) \neq 0$ for all $n$ and $u \in \mathbb{R}$ satisfying:
\begin{itemize}
    \item[\namedlabel{fourier:W}{C(3)}] The Fourier transforms of the kernel sequence $W_n$ are integrable with bounded integrals, i.e. $\int \left\vert W_n^{\dagger}\left(  u\right)\right\vert du < M < \infty$.
    
    \item[\namedlabel{fourier:Id*W}{C(4)}] Defining $\Xi_{n,p} (u) = W_n(u) u^p$, the Fourier transforms of $\Xi_{n,p}$ are integrable with bounded integrals, i.e. $\int \left\vert \Xi_{n,p}^{\dagger}\left(  u\right)\right\vert du < M< \infty$, for $p=1,2,3$.
    
    \item[\namedlabel{decay:W}{C(5)}]   The kernel functions $W_n(\cdot)$ are symmetric around zero and monotone non-increasing on $[1,+\infty)$.

    \item[\namedlabel{decay:Xi}{C(6)}] The functions  $\Xi_{n,p}(\cdot)$  are non-increasing on $[1, \infty)$ and the limiting objects  $\Xi_{n,p}(1^+) := \underset{u \searrow 1 }{\mathrm{lim}} \Xi_{n,p}(u)$ ($= W_n(1^+)$) tend to zero as $n$ tends to infinity, for $p=1,2,3$.
\end{itemize}
Where we are using $g^{ \dagger}(\cdot)$  to indicate the Fourier transform  $g^{ \dagger}(\cdot) = \frac{1}{2 \pi} \int e^{-\mymathbb{i}u \cdot}g(u)du$, for general function $g(\cdot)$, and $\mymathbb{i}$ to denote  the imaginary unit.
Kernels satisfying \eqref{fourier:W}-\eqref{decay:Xi} will still yield Theorems \ref{Thm:Ghat-G}
but with the presence of the additional terms  $O\left( h_G^{-2} \right)   W_n\left( 1^{+}\right) $ 
in the right hand side of  \eqref{bound:Ghat-G}.
To remove this additional term, one needs to ensure that this term is dominated by the remaining terms, and this requires specifying the kernel more concretely. Our special choices of  $ \mathbf{W}_n(u)$ (equation \eqref{kernel1}) or  $\mathbb{W}_n(u)$ (equation \eqref{kernel2}) do just that: these are kernels that satisfy \ref{fourier:W}--\ref{decay:Xi} while also ensuring that the extra term $O\left( h_G^{-2} \right)   W_n\left( 1^{+}\right) $ 
is dominated by the remaining terms.
\bibliography{CovEstimation}
\end{document}